\documentclass[10pt,letterpaper,amsmath,amssymb]{article}
\usepackage[english]{babel}
\usepackage{a4wide}
\usepackage{amsmath,amsthm}
\usepackage[numbers,sort]{natbib}
\usepackage{braket}
\usepackage{bm}
\usepackage{graphicx}
\usepackage{subcaption}

\theoremstyle{plain}
\newtheorem{thm}{Theorem}[section]

\theoremstyle{definition}

\theoremstyle{remark}

\newcommand{\average}[1]{\ensuremath{\langle#1\rangle} }

\begin{document}
\baselineskip 18pt
\begin{center}
{\Large {\bf Truncated Wigner function theory of coherent Ising machines based on degenerate optical parametric oscillator network}}
 \vskip5mm Daiki Maruo\footnote{Email address: maruo@nii.ac.jp}$^{,2}$,
Shoko Utsunomiya\footnote{Email address: shoko@nii.ac.jp}, Yoshihisa Yamamoto$^{2,}$\footnote{Email address: yyamamoto@stanford.edu}\vskip3mm \mbox{}%

$^{1}$ Graduate School of Information Science and Technology, The University of Tokyo 7-3-1 Bunkyo-ku, Tokyo, 113-8654

\mbox{}$^{2}$ National Institute of Informatics, 2-1-2 Chiyoda-ku, Tokyo, 102-0076

\mbox{}$^{3}$ E.L. Ginzton Laboratory, Stanford University, Stanford, CA 9430

\bigskip

\begin{abstract}
We present the quantum theory of coherent Ising machines based on networks of degenerate optical parametric oscillators (DOPOs). 
In a simple model consisting of two coupled DOPOs, both positive-$P$ representation and truncated Wigner representation predict quantum correlation and inseparability between the two DOPOs in spite of the open-dissipative nature of the system.  
Here, we apply the truncated Wigner representation method to coherent Ising machines with thermal, vacuum, and squeezed reservoir fields. We find that the probability of finding the ground state of a one-dimensional Ising model increases substantially as a result of reducing excess thermal noise and squeezing the incident vacuum fluctuation on the out-coupling port.
\bigskip
\end{abstract}
\end{center}

\newpage
\section{Introduction}
A degenerate optical parametric oscillator (DOPO) provides a simple and clean experimental platform for investigating various quantum and  coherent effects, such as squeezing \cite{walls1983squeezed}\cite{PhysRevLett.57.2520}, entanglement \cite{PhysRevLett.60.2731}\cite{janousek2009optical}, quantum teleportation \cite{bouwmeester1997experimental}\cite{furusawa1998unconditional}, frequency combs \cite{Wong:10}, coherent feedback control \cite{Crisafulli:13}, and quantum information processing \cite{yokoyama2013ultra}.
\par
We have proposed and demonstrated a novel computing system, called a coherent Ising machine (CIM), based on a network of mutually coupled DOPOs \cite{PhysRevA.88.063853}\cite{marandi2014network}.
Each Ising spin is emulated by the in-phase amplitude $\hat{x}$ of the DOPO, which takes either $0$-phase (up-spin) or $\pi$-phase (down-spin). Note that at a pump rate well above the oscillation threshold, one Ising spin is collectively represented with many photons in each DOPO. 
In a CIM at the threshold pump rate, however, each DOPO has only one or a few photons and preserves a coherent superposition of the binary $0$-phase and $\pi$-phase in spite of its dissipative coupling to external reservoirs \cite{wolinsky1988quantum}. 
''Which-path'' information for the $0$-phase vs. $\pi$-phase is buried within the increased quantum noise in the in-phase (anti-squeezed) amplitude $\hat{x}$, which allows superposition of the two states\cite{krippner1994transient}. 
If two DOPOs are mutually coupled via a common optical path, quantum correlation and entanglement form between them in a broad pumping range below the threshold \cite{Kenta}.
The CIM eventually self-stabilizes one particular ground state above the oscillation threshold via spontaneous symmetry breaking associated with the second order phase transition.
\par
In order to implement a large-scale CIM, we need to prepare $N$ identical DOPOs and connect them with $N^2$ optical coupling paths. In most general case of full and asymmetric connection machines, we need to implement $N^2$ Ising coupling constants $J_{ij} \neq J_{ji}$. This is a daunting task if we consider a problem size of $N=10^{3}$ or larger. A CIM with a fiber ring resonator and optical delay lines, shown in Fig.~\ref{fig:CIM}, has been proposed as a practical means of implementing the concept \cite{haribara2016computational}. 
In this configuration, $N$ independent DOPO pulses are simultaneously generated in a common ring resonator with a pump laser pulse train. The round trip time of the cavity is adjusted $N$ times to the DOPO pulse interval. By providing appropriate delay to the optical coupling pulse and injecting it into the fiber ring resonator at an appropriate timing, the Ising model is mapped to the total photon loss of the DOPO network \cite{PhysRevA.88.063853}\cite{marandi2014network}.   
When the pump laser power is gradually increased from below to above the oscillation threshold, the DOPO network oscillates with the phase configuration having the minimum network loss, and thus, the resulting phase configuration of the spontaneously selected oscillation mode corresponds to the ground state of the given Ising Hamiltonian. In this optical coupling scheme, the amplitudes of each coupling pulse include random noise arising from the pump field and incident vacuum fluctuations from the open port of the output coupler (See Fig.~1). 
By squeezing the vacuum fluctuations along the in-phase amplitude $\hat{x}$ with a phase sensitive amplifier \cite{walls1983squeezed}\cite{PhysRevLett.57.2520} at this open port, couplings can be achieved with reduced perturbations and improved signal-to-noise ratios.
The cost of this non-classical coupling strategy is the noise added to the quadrature-phase amplitude $\hat{p}$ of the internal DOPO pulse; however, such noise does not affect the operation of the CIM, since the quadrature-phase amplitude is deamplified by the internal phase sensitive amplifier and also attenuated by the internal linear loss. The use of the squeezed reservoir field also helps to maintain the coherent superposition of $0-phase$ and $\pi-phase$ states \cite{munro1995transient}.
This is the particular scheme that we will study in this paper.
\par
This paper is organized as follows. In Sec.~II, we describe the c-number stochastic differential equations (CSDEs) derived using the positive-$P$ representation and truncated Wigner representation methods. 
In Sec.~III, we study the quantum correlation and inseparability in the two coupled DOPO systems. We show that the two methods predict the same amount of inseparability between the two DOPOs. 
In Sec.~IV, we examine the behavior of the one-dimensional ring with 16 DOPOs from the viewpoint of the success probability  statistics. Through this analysis, we identify  four computational stages of the CIM: quantum parallel search, quantum filtering, spontaneous symmetry breaking, and quantum-to-classical crossover. 
In Sec.~V, we introduce the realistic model of a CIM composed of  discrete devices, including degenerate optical parametric amplifiers (DOPAs),  fiber, output couplers and injection couplers. 
Finally in Sec.~VI, we demonstrate how reducing excess thermal noise and even squeezing vacuum fluctuations at the out-coupling port improves the probability of finding the ground state of a one-dimensional Ising model with 16 spins.
\newpage
\section{C-number stochastic differential equations for coupled DOPOs}
First, we consider a simple model shown in Fig.~\ref{fig:2DOPO} as a building block for the CIM. 
This model is equivalent to the coherent Ising machine with optical delay line coupling (Fig.~\ref{fig:CIM}) when we set the gain of PSA $G=1$.
Two DOPOs are optically coupled via a common optical path, which is further coupled to external reservoirs. 
The total Hamiltionian of the system is \cite{Kenta}
\begin{equation}
\label{eq:H}
\mathcal{H} = \mathcal{H}_{free} + \mathcal{H}_{int} + \mathcal{H}_{pump} + \mathcal{H}_{mirror} + \mathcal{H}_{SR} \\
\end{equation}
\begin{equation}
\label{eq:H_free}
\mathcal{H}_{free} = \hbar \omega_{s} \sum_{j = 1}^{2} \hat{a}_{sj}^{\dagger} \hat{a}_{sj} + \hbar \omega_{p} \sum_{j = 1}^{2} \hat{a}_{pj}^{\dagger} \hat{a}_{pj} + \hbar \omega_{s} \hat{a}_{c}^{\dagger} \hat{a}_{c} \\
\end{equation}
\begin{equation}
\label{eq:H_int}
\mathcal{H}_{int} = \frac{i\hbar\kappa}{2} \sum_{j = 1}^{2} (\hat{a}_{sj}^{\dagger2} \hat{a}_{pj} - \hat{a}_{pj} ^{\dagger} \hat{a}_{sj}^{2}) \\
\end{equation}
\begin{equation}
\label{eq:H_pump}
\mathcal{H}_{pump} = i \hbar \sum_{j=1}^{2} (\epsilon \hat{a}_{pj}^{\dagger} e^{-i\omega_{d}t} - \epsilon \hat{a}_{pj} e^{i\omega_{d}t}) \\
\end{equation}
\begin{equation}
\label{eq:H_mirror}
\mathcal{H}_{mirror} = i \hbar \zeta (\hat{a}_{c} \hat{a}_{s1}^{\dagger} - \hat{a}_{c}^ {\dagger} \hat{a}_{s1} + \hat{a}_{s2}\hat{a}_{c}^{\dagger}e^{-ik_cz} - \hat{a}_{s2}^{\dagger}\hat{a}_{c}e^{ik_cz}) \\
\end{equation}
\begin{equation}
\label{eq:H_SR}
\mathcal{H}_{SR} = \hbar \sum_{j=1}^{2}( \hat{a}_{sj} \hat{\Gamma}_{Rsj}^{\dagger} + \hat{\Gamma}_{Rsj}\hat{a}_{sj}^{\dagger} + \hat{a}_{pj} \hat{\Gamma}_{Rpj}^{\dagger} + \hat{\Gamma}_{Rpj} \hat{a}_{pj}^{\dagger} ) + \hbar (\hat{a}_{c} \hat{\Gamma}^{\dagger}_{Rc} + \hat{\Gamma}_{Rc} \hat{a}_{c}^{\dagger} ) \\
\end{equation}
where $H_{free}$ is the free field Hamiltonian for the signal, pump, and central coupling field, $H_{int}$ is the parametric interaction Hamiltonian, $H_{pump}$ is the external pumping Hamiltonian, $H_{mirror}$ is the coupling Hamiltonian among the two DOPOs and the central optical path, and $H_{SR}$ is the system (signal, pump and the central coupling field) -reservoir interaction Hamiltonian. 
In eq.~(\ref{eq:H_mirror}), the phase factors of the central coupling mode at the facet of DOPO$_{2}$ are expressed as $\hat{a}_{c}^{\dagger}e^{-ik_cz}$ and $\hat{a}_{c}e^{ik_cz}$, where $k_c$ is the wavenumber of the central coupling mode and $z$ is the central path length. 
The ferromagnetic coupling and anti-ferromagnetic coupling are realized when $e^{ik_cz} = e^{-ik_cz} = 1$ and $e^{ik_cz} = e^{-ik_cz} = -1$, respectively. 
\par 
The standard technique \cite{carmichael1999statistical} allows us to obtain the master equation for the fields in two signal modes, two pump modes, and one central coupling mode. 
We then use the positive-$P$ representation $P(\bm{\alpha},\bm{\beta})$\cite{PPrep} for the five modes to expand the total field density operator $\rho$:
\begin{equation}
\label{eq:OD_op}
\rho = \int P(\bm{\alpha},\bm{\beta}) \frac{\ket{\bm{\alpha}}\bra{\bm{\beta}}}{\braket{\bm{\beta}^*|\bm{\alpha}}}d\bm{\alpha}d\bm{\beta} \\
\end{equation}
where $\bm{\alpha} = (\alpha_{s1},\alpha_{s2},\alpha_{p1},\alpha_{p2},\alpha_{c})^{\rm{T}}$ and $\bm{\beta} = (\beta_{s1},\beta_{s2},\beta_{p1},\beta_{p2},\beta_{c})^{\rm{T}}$ are each expressed in terms of five complex numbers, and $\ket{\bm{\alpha}} = \ket{\alpha_{s1}}\ket{\alpha_{s2}}\ket{\alpha_{p1}}\ket{\alpha_{p2}}\ket{\alpha_{c}}$ and $\ket{\bm{\beta}} = \ket{\beta_{s1}}\ket{\beta_{s2}}\ket{\beta_{p1}}\ket{\beta_{p2}}\ket{\beta_{c}}$ are the multimode coherent states \cite{PhysRev.131.2766}. 
Here, $\alpha_{X}$ and $\beta_{X}$ are statistically independent, but their ensemble averaged excitation amplitudes satisfy $\average{\alpha_{X}} = \average{\beta_{X}^{*}}$, where $X$ denotes $s1, s2, p1, p2,$ and $c$ . 
We substitute (\ref{eq:OD_op}) into the master equation to obtain the Fokker-Planck equation for the distribution $P(\bm{\alpha},\bm{\beta}) $ \cite{PPrep}. The Ito rule governs the correspondence between the Fokker-Planck equation and the complex-number stochastic differential equation (CSDE) \cite{walls2007quantum}.
We can reach a series of CSDEs for the ten c-number variables $(\alpha_{s1},\beta_{s1})$, $(\alpha_{s2},\beta_{s2})$, $(\alpha_{p1},\beta_{p1})$, $(\alpha_{p2},\beta_{p2})$ and $(\alpha_{c},\beta_{c})$\cite{Kenta}:
\begin{eqnarray} \nonumber
d\alpha_{s1} &=& (-\gamma_s \alpha_{s1} + \kappa \alpha_{p1} \beta_{s1} +\zeta \alpha_c)dt + \sqrt{\kappa \alpha_{p1}}dw_{\alpha s1}(t) \\ \nonumber
d\beta_{s1} &=& (-\gamma_s \beta_{s1} + \kappa \beta_{p1} \alpha_{s1} + \zeta \beta_c)dt + \sqrt{\kappa \beta_{p1}}dw_{\beta s1}(t) \\ \nonumber
d\alpha_{s2} &=& (-\gamma_s \alpha_{s2} + \kappa \alpha_{p2} \beta_{s2} - \zeta e^{-ik_cz} \alpha_c)dt + \sqrt{\kappa \alpha_{p2}}dw_{\alpha s2}(t) \\ \nonumber
d\beta_{s2} &=& (-\gamma_s \beta_{s1} + \kappa \beta_{p2} \alpha_{s2} - \zeta e^{ik_cz} \beta_c)dt + \sqrt{\kappa \beta_{p2}}dw_{\beta s2}(t) \\ \nonumber
d\alpha_{p1} &=& (-\gamma_p \alpha_{p1} - \frac{\kappa}{2} \alpha_{s1}^2 + \epsilon )dt \\ \nonumber
d\beta_{p1} &=& (-\gamma_p \beta_{p1} - \frac{\kappa}{2} \beta_{s1}^2 + \epsilon )dt \\ \nonumber
d\alpha_{p2} &=& (-\gamma_p \alpha_{p2} - \frac{\kappa}{2} \alpha_{s2}^2 + \epsilon )dt \\ \nonumber
d\beta_{p2} &=& (-\gamma_p \beta_{p2} - \frac{\kappa}{2} \beta_{s2}^2 + \epsilon )dt \\ \nonumber
d\alpha_c &=& ( -\gamma_c \alpha_{c} - \zeta\alpha_{s1} + \zeta e^{ik_cz} \alpha_{s2} )dt \\
d\beta_c &=& ( -\gamma_c \beta_{c} - \zeta\beta_{s1} + \zeta e^{-ik_cz} \beta_{s2} )dt
\end{eqnarray}
Alternatively, we can expand the field density operator $\rho$ by using the Wigner function $W(\bm{\alpha}) $\cite{walls2007quantum}:
\begin{equation}
\label{eq:OD-op}
\rho = \int e^{\bm{\lambda}^* \bm{\hat{a}} - \bm{\lambda} \bm{\hat{a}}^{\dagger} } \left\{ \int e^{\bm{\lambda} \bm{\alpha}^* - \bm{\lambda}^* \bm{\alpha}} W(\bm{\alpha}) d\bm{\alpha} \right\} d\bm{\lambda}, \\
\end{equation}
where $\bm{\hat{a}} = (\hat{a}_{s1},\hat{a}_{s2},\hat{a}_{p1},\hat{a}_{p2},\hat{a}_c)^{\rm{T}}$ and $\bm{\lambda} = (\lambda_{s1},\lambda_{s2},\lambda_{p1},\lambda_{p2},\lambda_{c})$. 
$\bm{\alpha}$ and $\bm{\lambda}$ form a pair of complex numbers related by the Fourier transform: $\chi(\bm{\lambda}) = \int e^{\bm{\lambda} \bm{\alpha}^* - \bm{\lambda}^* \bm{\alpha}} W(\bm{\alpha}) d\bm{\alpha} $, where $\chi({\bm{\lambda})}$ is the symmetric correlation function \cite{walls2007quantum}. 
The resulting Fokker-Planck equation with the third and higher-order terms truncated gives another set of CSDEs:
\begin{eqnarray}
\label{eq:SDE}
\nonumber
d\alpha_{s1} &=& (-\gamma_s \alpha_{s1} + \kappa \alpha_{p1} \alpha_{s1}^{*} +\zeta \alpha_c)dt + \sqrt{\gamma_s}dW_{s1}(t) \\ \nonumber
d\alpha_{s2} &=& (-\gamma_s \alpha_{s2} + \kappa \alpha_{p2} \alpha_{s2}^{*} -\zeta e^{-ik_cz} \alpha_c )dt + \sqrt{\gamma_s}dW_{s2}(t) \\ \nonumber
d\alpha_{p1} &=& (-\gamma_p \alpha_{p1} - \frac{\kappa}{2} \alpha_{s1}^2 + \epsilon )dt + \sqrt{\gamma_p}dW_{p1}(t) \\ \nonumber
d\alpha_{p2} &=& (-\gamma_p \alpha_{p2} - \frac{\kappa}{2} \alpha_{s2}^2 + \epsilon )dt + \sqrt{\gamma_p}dW_{p2}(t) \\ 
d\alpha_c &=& ( -\gamma_c \alpha_{c} - \zeta\alpha_{s1} + \zeta e^{ik_cz} \alpha_{s2} )dt +\sqrt{\gamma_c}dW_{c}(t)
\end{eqnarray}
Here, $dW_{X}(t)$ is the c-number Wiener process and corresponds to the noise term in the equivalent Langevin equations. Next, we assume $\gamma_p,\gamma_c \gg \gamma_s$ and adiabatically eliminate the pump and central coupling modes $(d\alpha_{pj} = d\alpha_{c} = 0)$. We also assume $e^{ik_cz} = e^{-ik_cz} = -1$ (anti-ferromagnetic coupling). 
Finally, we obtain the CSDE for the normalized signal amplitude:
\begin{eqnarray}
\label{eq:SignalSDE}
\nonumber
dA_{s1} &=& \bigl\{ -A_{s1} + (E - A_{s1}^{2})A_{s1}^{*} - \xi A_{s2} \bigr\} d\tau + g dW'_{s1}(\tau)\\
dA_{s2} &=& \bigl\{ -A_{s2} + (E - A_{s2}^{2})A_{s2}^{*} - \xi A_{s1} \bigr\} d\tau + g dW'_{s2}(\tau)
\end{eqnarray}
where $A_{sj} = g\alpha_{sj}$ is the normalized signal amplitude, $g = \frac{\kappa}{\sqrt{2\gamma'_s\gamma_p}}$ is the saturation parameter, $\gamma'_s = \gamma_s + \frac{\zeta^2}{\gamma_c}$ is the effective signal field decay rate, $E = \frac{\kappa}{\gamma'_s\gamma_p}\epsilon$ is the normalized pump rate, $\tau = \gamma'_st$ is the normalized time, and $\xi = \frac{\zeta^2}{\gamma_s\gamma_c + \zeta^2} = \frac{\zeta^2}{\gamma'_s\gamma_c}$ is the normalized effective coupling constant \cite{Kenta}. 
The noise term $dW'_{s1}$ and $dW'_{s2}$ is:
\begin{eqnarray}
\label{eq:SignalSDE_NOISE}
\nonumber
dW'_{s1} &=& \sqrt{\frac{\gamma_s}{\gamma'_s}} dW_{s1}(\tau) + A_{s1}dW_{p1}(\tau) + \sqrt{\xi} dW_{c}(\tau)\\
dW'_{s2} &=& \sqrt{\frac{\gamma_s}{\gamma'_s}} dW_{s2}(\tau) + A_{s2}dW_{p2}(\tau) + \sqrt{\xi} dW_{c}(\tau)
\end{eqnarray}
The equivalent CSDE for the positive-$P$ representation can be found in eqs.~(25) and (26) of ref.~\cite{Kenta}.
\par
We can easily extend this results to the one-dimensional ring network consisting of $N$ DOPOs.
Figure \ref{fig:ring} is the sketch of one-dimensional ring network consisting of $N=16$ DOPOs.
The CSDE of $j$th DOPO constructing one-dimensional ring network is
\begin{eqnarray}
\label{eq:SignalSDE_1Dring}
\nonumber
dA_{sj} &=& \bigl\{ -A_{sj} + (E - A_{sj}^{2})A_{sj}^{*} - \xi A_{sj-1} -\xi A_{sj+1} \bigr\} d\tau + g dW'_{sj}(\tau) \\
dW'_{sj} &=& \sqrt{\frac{\gamma_s}{\gamma'_s}} dW_{sj}(\tau) + A_{sj}dW_{pj}(\tau) + \sqrt{\xi} dW_{cj+1}(\tau) + \sqrt{\xi} dW_{cj}(\tau)
\end{eqnarray}
In the one-dimensional ring network case, we employ the boundary condition as periodic, i.e. $\hat{x}_{s(j+N)} = \hat{x}_{sj}$, $\hat{p}_{s(j+N)} = \hat{p}_{sj}$ and $dW_{c(j+N)} = dW_{cj}$ are satisfied.
\newpage
\section{Quantum correlation and inseparability}
The expectation value of a normally ordered operator is readily evaluated using the positive-$P$ function \cite{PPrep}:
\begin{equation}
\label{eq:OD-op}
\average{\hat{a}_{s1}^{\dagger j} \hat{a}_{s2}^{\dagger k} \hat{a}_{s1}^l \hat{a}_{s2}^m} = \int \beta^{j}_{s1}\beta^{k}_{s2}\alpha^{l}_{s1}\alpha^{m}_{s2} P(\left\{ \alpha \right\},\left\{ \beta \right\})d\bm{\alpha}d\bm{\beta}, \\
\end{equation}
while the expectation value of a symmetrically ordered operator is conveniently evaluated using the truncated Wigner function \cite{walls2007quantum}:
\begin{equation}
\label{eq:SD-op}
\average{\hat{a}_{s1}^{\dagger j} \hat{a}_{s2}^{\dagger k} \hat{a}_{s1}^l \hat{a}_{s2}^m}_{\rm{S}} = \int \alpha^{*j}_{s1}\alpha^{*k}_{s2}\alpha^{l}_{s1}\alpha^{m}_{s2} W(\left\{ \alpha \right\})d\bm{\alpha}. \\
\end{equation}
Here, $\left\{ \alpha \right\} = (\alpha_{s1}, \alpha_{s2})^{\rm{T}}$ and $\left\{ \beta \right\} = (\beta_{s1}, \beta_{s2})^{\rm{T}}$ for the case of two coupled DOPOs.
The correlation function between the two DOPOs for the in-phase amplitude $\hat{x} = (\hat{a} + \hat{a}^{\dagger})/2$ and quadrature-phase amplitude $\hat{p} = (\hat{a} - \hat{a^{\dagger}})/(2i)$ is defined as
\begin{eqnarray}
\nonumber
\label{eq:COR}
\mathrm{C}(\hat{x}_{s1},\hat{x}_{s2}) &=& \frac{\average{\hat{x}_{s1} \hat{x}_{s2}}}{\average{\Delta x_{s1} \Delta x_{s2} }} = \frac{\average{c_{s1}c_{s2}}}{ \sqrt{\average{c_{s1}^2} - \average{c_{s1}}^2} \sqrt{\average{c_{s2}^2} - \average{c_{s2}}^2}} ,\\
\mathrm{C}(\hat{p}_{s1},\hat{p}_{s2}) &=& \frac{\average{\hat{p}_{s1} \hat{p}_{s2}}}{\average{\Delta p_{s1} \Delta p_{s2}} } = \frac{\average{s_{s1}s_{s2}}}{ \sqrt{\average{s_{s1}^2} - \average{s_{s1}}^2} \sqrt{\average{s_{s2}^2} - \average{s_{s2}}^2}} ,
\end{eqnarray}
where $c_{X} = (\alpha_{X} + \alpha_{X}^{*})/2$, $s_{X} = (\alpha_{X} - \alpha_{X}^{*})/(2i)$, and $\Delta O = \sqrt{\average{\hat{O}^2} - \average{\hat{O}}^2}$ for a general operator $\hat{O}$. 
We will use the EPR-type operators $\hat{u}_{+} = \hat{x}_{s1} + \hat{x}_{s2}$ and $\hat{v}_{-} = \hat{p}_{s1} - \hat{p}_{s2}$ to evaluate the quantum correlation and entanglement. 
We assume that the two DOPOs are coupled with the anti-ferromagnetic phase, i.e.~$e^{ik_cz} = e^{-ik_cz} = -1$, so that we expect that $\hat{x}_{s1}$ and $\hat{x}_{s2}$ are negatively correlated, while $\hat{p}_{s1}$ and $\hat{p}_{s2}$ are positively correlated. 
The condition for negative (or positive) quantum correlation is given by $\average{\Delta\hat{u}_{+}^2} < 0.5$ or $\average{\Delta\hat{v}_{-}^2} < 0.5$, while the criterion for inseparability is given by $\average{\Delta\hat{u}_{+}^2} + \average{\Delta\hat{v}_{-}^2} < 1 $\cite{PhysRevLett.84.2722}.
Figure~\ref{fig:ent} compares the total variances of the EPR-type operator, computed by the positive-$P$ representation and by the truncated Wigner representation. 
Here, the pump rate gradually and linearly increases from zero to 1.5 times the oscillation threshold over time $\tau = 200$, i.e.~$E = 1.5(\tau/200)$. The saturation parameter is $g = 0.01$. 
As can be seen in Fig.~\ref{fig:ent}, the two coupled DOPOs feature inseparability, i.e., $\average{\Delta\hat{u}_{+}^2} + \average{\Delta\hat{v}_{-}^2} \leq 1$, when the system evolves from below to above the oscillation threshold. 
Note that the coupled DOPO threshold pump rate is given by $E_{th} = 1-\xi$ rather than $E^{(0)}_{th} = 1$ for a solitary (uncoupled) DOPO \cite{PhysRevA.88.063853}. 
Increasing the coupling constant $\xi$ enhances the inseparability. As expected from the previous study for a solitary DOPO \cite{drummond2002critical}\cite{chaturvedi2002limits}\cite{dechoum2004critical}, the results obtained using the positive-$P$ representation are indistinguishable from those of the truncated Wigner representation. 
Our numerical simulation confirms that the difference in the total variances, $\average {\Delta \hat{u}_{+}^2} + \average{\Delta \hat{v}_{-}^2}$, evaluated using the positive-$P$ representation and the truncated Wigner representation is within the statistical error due to the finite number of sample functions $N = 200,000$, which is shown only in Fig.~\ref{fig:ent}(a) as vertical bars.
There is a variance spike from $\tau = 60$ to $ \tau = 80$ as the parameter is $\xi = 0.6$ in Fig.~\ref{fig:ent}.
We confirmed that the spike is due to the turn-on-delay oscillation effect because the spike disappears with slower gradual pumping.
We show the turn-on delay oscillation effect and the disappearance of the variance spike in Fig.~\ref{fig:noneq}, where the pumping schedule is varied from $E = 1.5(\tau/200),E = 1.5(\tau/400)$ to $E = 1.5(\tau/800)$.
%
%
%
%
%
%
%
%
%
\par
In the one-dimensional ring network case, we can define the operators $\hat{u}_{1D}$ and $\hat{v}_{1D}$ as the indicator of quantum correlation and inseparability if $N$, which is the number of DOPOs, is even.
The mathematical definition of $\hat{u}_{1D}$ and $\hat{v}_{1D}$ is
\begin{equation}
\hat{u}_{1D} = \sum_{j=1}^{N}\hat{x}_{sj}, \hat{v}_{1D} = \sum_{j = 1}^{N}(-1)^{j}\hat{p}_{sj}
\end{equation}
We made a proof below that the operator $\hat{u}_{1D}$ and $\hat{v}_{1D}$ are the indicators of the quantum correlation and inseperability of the system just like between two continuous variables \cite{PhysRevLett.84.2722}.
\begin{thm}
If the system is separable, the inequality $\average{\Delta \hat{u}_{1D}^2} + \average{\Delta \hat{v}_{1D}^2} \geq N/2$ is satisfied.
\end{thm}
\begin{proof}
The left-hand side of inequality, $\average{\Delta \hat{u}^2_{1D}} + \average{\Delta \hat{v}^2_{1D}}$ can be written as 
\begin{equation}
\label{eq:ineq_1D}
\average{\Delta \hat{u}^2_{1D}} + \average{\Delta \hat{v}^2_{1D}} = \mathrm{Tr}[\hat{\rho} \hat{u}^2_{1D}] - (\mathrm{Tr}[\hat{\rho} \hat{u}_{1D}])^2 + \mathrm{Tr}[\hat{\rho} \hat{v}^2_{1D}] - (\mathrm{Tr}[\hat{\rho} \hat{v}_{1D}])^2.
\end{equation}

If the system is separable,  the system density operator $\hat{\rho}$ can be decomposed as the tensor product of density operator of each DOPO $\rho_{jk}$, i.e.~
\begin{equation}
\hat{\rho} = \sum_{k}q_{k} \hat{\rho}_{1k} \otimes \hat{\rho}_{2k}  \otimes \cdots \otimes \hat{\rho}_{Nk} = \sum_{k} q_{k} \prod_{j=1}^{N} \hat{\rho}_{jk}
\end{equation}
Here $q_{k}$ is the mixing probability of each tensor product $\prod_{j=1}^{N} \hat{\rho}_{jk}$ and $\sum_{k}q_{k} = 1$ is satisfied.
Each term of (\ref{eq:ineq_1D}) can be written as
\begin{equation}
\begin{split}
\mathrm{Tr}[\hat{\rho} \hat{u}^2_{1D}]  &= \mathrm{Tr}[\sum_{k} q_{k} \prod_{j=1}^{N} \hat{\rho}_{jk} \hat{u}_{1D}^2] = \mathrm{Tr}[\sum_{k} q_{k} \prod_{j=1}^{N} \hat{\rho}_{jk}  \sum_{j=1}^{N} (\hat{x}_{sj}^2 + 2\sum_{l=j+1}^{N} \hat{x}_{sj}\hat{x}_{sl})  ] \\
& = \sum_{k}q_{k}  \sum_{j = 1}^{N} \{ \average{\hat{x}_{sj}^2}_{k} + 2\sum_{l=j+1}^{N}\average{\hat{x}_{sj} }_{k} \average{\hat{x}_{sl}}_{k} \} \\
&= \sum_{k}q_{k} \sum_{j = 1}^{N} \{ \average{\Delta \hat{x}_{sj}^2}_{k} + \average{\hat{x}_{sj}}_k^2 + 2\sum_{l=j+1}^{N}\average{\hat{x}_{sj}}_{k} \average{\hat{x}_{sl}}_{k} \} \\
&= \sum_{k}q_{k} \sum_{j = 1}^{N} \average{\Delta \hat{x}_{sj}^2}_{k} + \sum_{k}q_{k} \average{\hat{u}_{1D}}^2_{k}\\
\end{split}
\end{equation}
\begin{equation}
\begin{split}
(\mathrm{Tr}[\hat{\rho} \hat{u}_{1D}])^2 &=  (\mathrm{Tr}[\sum_{k} q_{k} \prod_{j=1}^{N} \hat{\rho}_{jk} \hat{u}_{1D}])^2 = ( \sum_{k}q_{k} \average{\hat{u}_{1D}}_{k})^2 \\
\end{split}
\end{equation}
\begin{equation}
\begin{split}
\mathrm{Tr}[\hat{\rho} \hat{v}^2_{1D}]  &= \mathrm{Tr}[\sum_{k} q_{k} \prod_{j=1}^{N} \hat{\rho}_{jk} \hat{v}_{1D}^2] = \mathrm{Tr}[\sum_{k} q_{k} \prod_{j=1}^{N} \hat{\rho}_{jk}  \sum_{j=1}^{N} (\hat{p}_{sj}^2 + 2\sum_{l=j+1}^{N} (-1)^{j+l-2}\hat{p}_{sj}\hat{p}_{sl})  ] \\
& = \sum_{k}q_{k} \sum_{j = 1}^{N} \{\average{\hat{p}_{sj}^2}_{k} + 2\sum_{l=j+1}^{N}(-1)^{j+l-2}\average{\hat{p}_{sj} }_{k} \average{\hat{p}_{sl}}_{k} \} \\
&= \sum_{k}q_{k} \sum_{j = 1}^{N} \{ \average{\Delta \hat{p}_{sj}^2}_{k} + \average{\hat{p}_{sj}}_k^2 + 2\sum_{l=j+1}^{N}(-1)^{j+l-2}\average{\hat{p}_{sj}}_{k} \average{\hat{p}_{sl}}_{k} \} \\
&= \sum_{k}q_{k} \sum_{j = 1}^{N} \average{\Delta \hat{p}_{sj}^2}_{k} + \sum_{k}q_{k}\average{\hat{v}_{1D}}^2_{k} \\
\end{split}
\end{equation}
\begin{equation}
\begin{split}
(\mathrm{Tr}[\hat{\rho} \hat{v}_{1D}])^2 &=  (\mathrm{Tr}[\sum_{k} q_{k} \prod_{j=1}^{N} \hat{\rho}_{jk} \hat{v}_{1D}])^2 = ( \sum_{k}q_{k} \average{\hat{v}_{1D}}_{k})^2
\end{split}
\end{equation}
Here we use the inequalities $\sum_{k}q_{k}\average{\hat{u}_{1D}}^2_{k} = \sum_{k}q_{k}\sum_{k}q_{k} \average{\hat{u}_{1D}}^2_{k} \geq (\sum_{k}q_{k} \average{\hat{u}_{1D}}_{k})^2$ and $\sum_{k}q_{k}\average{\hat{v}_{1D}}_{k}^2 = \sum_{k}q_{k}\sum_{k}q_{k} \average{\hat{v}_{1D}}_{k}^2 \geq (\sum_{k}q_{k} \average{\hat{v}_{1D}}_{k})^2$, which are derived from Cauthy-Schwarz inequality.
Then, $\sum_{k} q_{k} (\average{\Delta \hat{x}^2_{sj}}_{k} + \average{\Delta \hat{p}^2_{sj}}_{k}) \geq 0.5$ is derived from the uncertainty principle.
We can conclude that the inequality $\average{\Delta \hat{u}_{1D}^2}+\average{\Delta \hat{v}_{1D}^2} \geq N/2$ is satisfied if the state is separable.
\end{proof}
\begin{thm}
If the inequality $\average{\Delta \hat{u}_{1D}^2} + \average{\Delta \hat{v}_{1D}^2} < N/2$ is satisfied, the system is inseparable and the quantum correlation exists.
\end{thm}
\begin{proof}
It is the contraposition of Theorem 3.1.
\end{proof}
Figure~\ref{fig:insep} shows the total variances of the EPR-like  operator $\hat{u}_{1D} + \hat{v}_{1D}$ for various squeezing parameters $r$ of the input states into the output coupler in the $N=16$ one-dimensional ring network.
Here, the pump rate gradually increases from zero to 0.375 times the oscillation threshold over time $\tau = 200$, i.e.~$E = 0.375(\tau/200)$. The saturation parameter is $g = 0.01$ and the coupling constant $\xi = 0.4$. 
If a standard vacuum fluctation is incident on the output coupler, the quantum correlation exists in the quadrature-phase amplitudes $( \average{\Delta \hat{v}_{1D}^2} < N/4)$ but only classical correlation exists in the in-phase amplitudes$( \average{\Delta \hat{u}_{1D}^2} \ge N/4)$ \cite{Kenta}. 
On the other hand,
if we inject squeezed vacuum states with reduced quantum noise, $e^{-2r}/4$, in the in-phase amplitude and enhanced quantum noise, $e^{2r}/4$, in the quadrature -phase amplitude, the quantum correlation exists in both in-phase and quadrature-phase amplitudes, and is boosted with increasing the squeezing parameter.
The variance of $\hat{u}_{1D} + \hat{v}_{1D}$ is below the $ N/2=8 $, which is the criteria of inseperability and holds even without squeezing.
\par
In Fig.~\ref{fig:ent}-6, we assume relatively lange coupling constants $\xi(=0.4-0.995)$. Such a strong coupling is not unrealistic if we amplify the out-coupled field with a noiseless PSA as shown in Fig1.
\newpage
\section{Quantum parallel search, quantum filtering, spontaneous symmetry breaking and quantum-to-classical crossover}
When the normalized pump rate for each DOPO is increased linearly in time according to $E = 1.5(\tau/200)$, the inseparability forms over a wide pumping range from below to above the oscillation threshold, as shown in Fig.~\ref{fig:ent}. 
The static threshold for the coupled DOPOs is equal to $E_{th} = 1 - \xi = 0.4$, when $\xi = 0.6$, which corresponds to the normalized time $\tau \simeq 53$. 
The coupled DOPOs actually oscillate, however, at $\tau \simeq 80$  because of the turn-on delay effect. 
The inseparability emerges well below the dynamic oscillation threshold and weakens toward the dynamic threshold. Figures~\ref{fig:bar}(a), (b), (c), and (d) show the post-selected probabilities of obtaining the measurement results for four possible spin configurations, i.e., $\ket{\uparrow \uparrow}$, $\ket{\uparrow \downarrow}$, $\ket{\downarrow \uparrow}$ and $\ket{\downarrow \downarrow}$, at a specific time $\tau$ (or normalized pump rate $E$). 
The post-selection is performed under the condition that the final state is $\ket{\uparrow \downarrow}$ . At $\tau = 5$ $(E=0.0375)$, immediately after the pump power is switched on, where the average photon number of each DOPO is $\average{n} \simeq 0.01$, almost identical probabilities of $1/4$ are found for all four spin configurations. This result suggests the uncorrelated product state, $\frac{1}{\sqrt{2}}(\ket{\uparrow} + \ket{\downarrow})_{1} \bigotimes \frac{1}{\sqrt{2}}(\ket{\uparrow} + \ket{\downarrow})_{2}$, where the two DOPOs are independently in linear superposition states of the $0$-phase and $\pi$-phase. 
%
%
Here, the $0$-phase ($\ket{\uparrow}$) and $\pi$-phase ($\ket{\downarrow}$) are not legitimate orthogonal phase states,  but operationally defined by the equation:
\begin{eqnarray}
\label{eq:Gstate}
\nonumber
\ket{\rm{DOPO}} &=& c_{0}\ket{0} + c_{2}\ket{2} + c_{4}\ket{4} + c_{6}\ket{6} + \dotsb \\ \nonumber
&=&\frac{1}{2}(\!c_{0}\ket{0} \!+ \!c_{1}\ket{1} \!+ \!c_{2}\ket{2}\! +\! c_{3}\ket{3}\! +\! \dotsb\!) + \frac{1}{2}(\!c_{0}\ket{0}\! -\! c_{1}\ket{1}\! + \!c_{2}\ket{2}\! - \!c_{3}\ket{3}\! + \!\dotsb\!) \\ 
,
\end{eqnarray}
where the first term of the right hand side of the second line corresponds to $\ket{\uparrow}$ state while the second term corresponds to $\ket{\downarrow}$ state.
The almost equal probability of $1/4$ for all possible spin configurations imply that the system is prepared in a superposition of all possible states and has already started a ''quantum parallel search'' at this early stage. At $\tau = 20$ $(E = 0.15)$, where the average photon number of each DOPO is $\average{n} \simeq 0.5$, the probabilities of finding the two degenerate ground states $\ket{\uparrow \downarrow}$ and $\ket{\downarrow \uparrow}$ are already higher than those for the excited states $\ket{\uparrow \uparrow}$ and $\ket{\downarrow \downarrow}$, as shown in Fig.~\ref{fig:bar}(b). 
Note that the probability amplitudes for all possible spin states probe the network connection and amplify/deamplify the probability amplitudes of the ground states/excited states, even when the average photon number per DOPO is smaller than one.  
The system evolves from the product state, $\frac{1}{2}(\ket{\uparrow \uparrow} + \ket{\uparrow \downarrow} + \ket{\downarrow \uparrow} + \ket{\downarrow \downarrow})$, to the entangled state, $\frac{1}{\sqrt{2}}(\ket{\uparrow \downarrow} + \ket{\downarrow \uparrow})$, already in this weak excitation regime.  We call this amplification/deamplification process ''quantum filtering''. 
When the pump rate exceeds the oscillation threshold, $E = 0.45 > E_{th}^{0} = 0.4$, where the average photon number is $\average{n} \simeq 2$, the coupled DOPO network selects a particular final state $\ket{\uparrow \downarrow}$ rather than $\ket{\downarrow \uparrow}$ via spontaneous symmetry breaking, as shown in Fig.~\ref{fig:bar}(c). 
The true probabilities of obtaining the two degenerate ground states $\ket{\uparrow \downarrow}$ and $\ket{\downarrow \uparrow}$ are 50-50\%, but the particular result shown in Fig~\ref{fig:bar}(c) is post-selected by the final result of $\ket{\uparrow \downarrow}$. 
Finally, at the dynamic threshold, $\tau = 80$ $(E = 0.6)$, the probability of finding a final result $\ket{\uparrow \downarrow}$ becomes nearly 100\%, as shown in Fig~\ref{fig:bar}(d). 
This final stage, called ''quantum to classical crossover'', is made possible by the collapse of the state due to the large separation between the $0$-phase and $\pi$-phase state and also by the stimulated emission of coherent photons with a particular phase .
\par
Figures \ref{fig:post}(a) and (b) plot the time evolutions of the two probabilities for the selected state $\ket{\uparrow \downarrow}$ and unselected states $\ket{\downarrow \uparrow}$  versus normalized time for different squeezing parameters $r$ for the two input states into the central cavity (see Fig.~\ref{fig:2DOPO}).
The initial increase in the two probabilities at $0 < \tau < 30$ reflects amplification of the two probability amplitudes by ''quantum filtering'', while the subsequent increase and decrease in the probabilities at $30 < \tau < 80$ is an indication of spontaneous symmetry breaking. Finally, the deterministic result surfaces at $\tau \simeq 80$, as a result of ''quantum to classical crossover.''
\par
We extend the same analysis to one-dimensional ring consisting of 16DOPOs, in which we post-selected the trajectory and followed its time evolution toward the specific final result.
The condition of post-selection is that the final state is one of the ground states, $\ket{\uparrow \downarrow \uparrow \downarrow \uparrow \downarrow \uparrow \downarrow \uparrow \downarrow \uparrow \downarrow \uparrow \downarrow \uparrow \downarrow}$.
Figure \ref{fig:searching} plots the time evolutions of the two probabilities for the selected ground state $\ket{\uparrow \downarrow \uparrow \downarrow \uparrow \downarrow \uparrow \downarrow \uparrow \downarrow \uparrow \downarrow \uparrow \downarrow \uparrow \downarrow}$ and unselected ground state $\ket{\downarrow \uparrow \downarrow \uparrow \downarrow \uparrow \downarrow \uparrow \downarrow \uparrow \downarrow \uparrow \downarrow \uparrow \downarrow \uparrow}$ for various squeezing parameter $r$.
We can see each step of the quantum filtering, spontaneous symmetry breaking and quantum-to-classical crossover in Fig.~\ref{fig:searching}.
The equation used for this simulation is given by (\ref{eq:SignalSDE_1Dring}) and the numerical parameters are $\xi = 0.4, E = 0.375(\tau/200)$ and $g = 0.01$.
At $\tau = 200$, the photon number of a DOPO is $\average{n} \simeq$ 2000.
The probability of getting one ground state by a random guess is only $2/2^{16} \simeq 0.00305\%$, while it is amplied to $0.03-0.3\%$ by quantum filterling before the spontaneous symmetry breaking sets in. 
\newpage
\section{Discrete model for CIM with multiple DOPO pulses and optical delay lines}
Suppose that the signal loss in the degenerate optical parametric amplifier (DOPA) and in the fiber ring cavity is negligible compared with the out-coupling loss for the mutual coupling between DOPO pulses in Fig.~1. 
Then the time evolution for the pump and signal fields inside the DOPA can be expressed as the following (truncated-Wigner) CSDE:
\begin{eqnarray}
\label{eq:SDE-OPA}
d\alpha_{p} &=& (\epsilon - \gamma_p\alpha_p - \frac{\kappa}{2}\alpha_s^2)dt + \sqrt{\gamma_p} dW_p(t) \\
d\alpha_{s} &=& \kappa \alpha_s^* \alpha_p dt
\end{eqnarray}
where $\epsilon$ is the external pump rate, $\gamma_p$ is the pump field decay rate, and $dW_p(t)$ is the complex Wiener process \cite{walls2007quantum}. 
We assume that the pump field decay rate $\gamma_p$ is very large so that the pump field dynamics obeys signal field dynamics. Under this slaving principle, the CSDE for the $i$-th DOPO signal pulse in the cavity is expressed as
\begin{equation}
\label{eq:SDE-OPA}
d\alpha_{si,cav} = \frac{\kappa}{\gamma_p}\alpha_{si,cav}^* (\epsilon - \frac{\kappa}{2} \alpha_{si,cav}^2) dt + \frac{\kappa}{\sqrt{\gamma_p}}\alpha_{si,cav}^* dW_{pi}(t)
\end{equation}
where the subscript $i$ designates the $i$-th signal pulse. We can normalize the equation (\ref{eq:SDE-OPA}), as we have done already in (\ref{eq:SignalSDE}):
\begin{equation}
\label{eq:NormalPPLN}
dA_{si,cav} = (E - A^{2}_{si,cav})A_{si,cav} ^{*} dT +\sqrt{2} \mu A_{si,cav}^{*} dW(T)
\end{equation}
where $T = \eta t$ is the round trip number inside the cavity, $\eta$ is the number of round trips per second, $\mu = \frac{\kappa}{\sqrt{2 \gamma_p \eta}}$,$E = \frac{\kappa}{\gamma_p \eta}\epsilon$ and $A_{si,cav} = \mu \alpha_{si,cav}$. 
Here, we should point out the difference between $g$ and $\mu$. 
While the saturation parameter $g$ determines the DOPO threshold and the photon number above the threshold, $\mu$ does not. 
This is because, in the present model, all the losses of the DOPO network depend only on the output coupler for the optical delay lines for the mutual coupling so that $\mu$ can not by itself govern the threshold and the photon number above the threshold.
\par
The out-coupling port for the optical delay lines, shown in Fig.~\ref{fig:CIM}, has the following input-output relation:
\begin{equation}
\label{eq:SDEPBS}
\left[
 \begin{array}{c}
 \alpha_{si,out} \\
 \alpha_{si,cav(t_1 +0)} \\
 \end{array}
 \right]
= \left[
 \begin{array}{cc}
 \sqrt{T_p} & -\sqrt{1-T_p} \\
 \sqrt{1-T_p} & \sqrt{T_p} \\
 \end{array}
 \right]
\left[
 \begin{array}{c}
 \alpha_{si,cav(t_1 -0)} \\
 f_{i} \\
 \end{array}
 \right]
\end{equation}
where $T_p$ is the power transmission coefficient of the output coupler, $\alpha_{si,out}$ is the $i$-th out-coupled signal field, and $f_i$ is the noise field incident from the open port of the out-coupler (Fig~.\ref{fig:CIM}). The noise field $f_i$ is a zero-mean complex-number Gaussian random variable. 
We will consider three cases, i.e., thermal state, vacuum state and squeezed vacuum state, for the input noise field $f_i$. The following phase sensitive amplifier (PSA) in Fig.~\ref{fig:CIM} amplifies the in-phase amplitude of the out-coupled field without any additional noise \cite{Yuen:83}:
\begin{equation}
\label{eq:MSIG}
c_{si,out} = G{\rm Re}(\sqrt{T_p}\alpha_{si,cav} - \sqrt{1-T_p}f_i)
\end{equation}
The optical signal from the PSA preserves all the statistical properties of the out-coupled field at a macroscopic (classical) level, which is needed to be split into multiple delay lines and generate the optical feedback pulse in a coherent state $\ket{\alpha_{FB}}$, where $\alpha_{FB} = \frac{1}{\sqrt{T_i}} \sum_{j}\xi_{ij}\tilde{c}_{sj,out}$. 
Here, $\xi_{ij}$ is the coupling constant from the $j$-th signal pulse to the $i$-th signal pulse, $\tilde{c}_{sj,out} = \frac{c_{sj,out}}{\sqrt{T_p}}$, and $\frac{1}{\sqrt{T_i}}(\gg 1)$ is the overall amplification factor including the PSA gain and the beam splitter loss. 
When the feedback pulse is injected back into the main cavity and combined with the $i$-th signal field circulating inside it, the power transmission coefficient of the injection coupler, shown in Fig.~\ref{fig:CIM}, is set to $T_{i} (\ll 1)$. 
Therefore, the quantum noise of the coherent state $\ket{\alpha_{FB}}$ is nearly completely suppressed when the feedback pulse is combined with the $i$-th signal field:
\begin{equation}
\label{eq:SDEOBS}
\left[
 \begin{array}{c}
 \alpha_{si,cav(t_2 +0)} \\
 \alpha_{si,ref} \\
 \end{array}
 \right]
= \left[
 \begin{array}{cc}
 \sqrt{T_i} & \sqrt{1-T_i} \\
 -\sqrt{1-T_i} & \sqrt{T_i} \\
 \end{array}
 \right]
\left[
 \begin{array}{c}
 \alpha_{FB} \\
 \alpha_{si,cav(t_2-0)} \\
 \end{array}
 \right]
\end{equation}
Thus, the only noise source, which is important in the optical delay line coupling scheme, is the incident noise field $f_i$ from the open port of the out-coupler.
\par
In order to suppress the error induced by the noise field $f_{i}$, we can squeeze the vacuum fluctuation with another phase sensitive amplifier which is not shown in Fig.1. 
Since the Ising Hamiltonian is implemented in the in-phase amplitude $c_{si,cav}$, we only need to suppress the quantum noise of the real part of $f_{i}$ at the cost of the increased quantum noise in the imaginary part of $f_i$. This is the strategy of a quantum nondemolition(QND) (or back action evading:BAE) measurement of the in-phase amplitude of the electromagnetic field \cite{PhysRevLett.62.28}. 
In the following numerical study, the (truncated) Wigner distribution function can be obtained as an ensemble average over many trajectories generated by numerical integration of the CSDE. 
The correlation and success probability are computed from the resulting Wigner distribution function.
\newpage
\section{Numerical results}
We numerically studied a one-dimensional ring configuration consisting of 16 Ising spins with identical anti-ferromagnetic couplings implemented by multiple DOPOs and optical delay lines. 
The external pump rate $\epsilon$ was switched on abruptly at $t = 0$, and the time evolutions of various quantities were evaluated over the  period of 2000 round trips inside the cavity. 
Figure \ref{fig:pn} shows the average photon number $\average{n_{si}}$ vs~. the normalized pump rate $p = \epsilon/\epsilon_{th}$ after 2000 round trips, which is considered to be a steady-state photon number for each pump rate. 
We defined the oscillation threshold pump rate $\epsilon_{th}$ to be the value of $\epsilon$ maximizing $d {\rm log}(\average{n_{si}}) / d {\rm log}(\epsilon) $. 
When a squeezed vacuum state is injected from the open port of the out-coupler, the anti-squeezed component of the squeezed vacuum state carries a finite photon number, which is the reason why there is a finite photon number in the limit of $p = \epsilon/\epsilon_{th} \rightarrow 0$ for a finite squeezing parameter $r \neq 0$.
\par
Figure \ref{fig:dudu}(a) shows the minimum variance $\average{\Delta \hat{u}_{+}^2}$ between two neighboring spins over a period of 2000 round trips after the pump is switched on. 
Each data point is an ensemble average over 40,000 samples. 
Note that the convergence of the valiance is problematic when the pump rate is near the oscillation threshold . 
It is assumed that the spin-spin coupling constant is relatively small($\xi = -0.01$) and the out-coupling coefficient is relatively large($T_p = 0.1$). 
As shown in Fig.~\ref{fig:dudu}(a), a quantum correlation, ($\average{\Delta \hat{u}_{+}^2} < 0.5$), forms between two neighboring spins up to a pump rate of $p < 1.5$ if a strongly squeezed vacuum state with $r=1.2$ is injected. 
Such a coherent Ising machine with squeezed input states forms a transient quantum correlation in the course of the computation. This quantum correlation is established in the entire system. 
On the other hand, if a standard vacuum state is incident on the open port, the correlation remains in the classical regime ($\average{\Delta \hat{u}_{+}^2} \geq 0.5$). Figure \ref{fig:dudu}(b) shows the final variance $\average{\Delta \hat{u}_{+}^2}$ between two neighboring spins after 2000 round trips. 
When the pump rate is $0.1 \leq p < 1.05$, the quantum correlation survives in the entire course of computation. 
At such a low pump rate, the quantum-to-classical crossover is never complete, even after 2000 round trips, and a particular computation result surfaces through the projection property of the detection process. 
\par
Figure~\ref{fig:ans} plots the probability of finding the ground state, which is a one-dimensional anti-ferromagnetic order, versus the normalized pump rate $p$. 
The maximum success probability occurs at a pump rate just above the oscillation threshold for each squeezing parameter from $r = 0$ to $r = 1.2$. 
As expected, it increases as the squeezing parameter increases. 
\par
Figure \ref{fig:ans_TP} plots the probability of finding one of the two degenerate ground states versus the normalized pump rate $p$ for different average thermal photon numbers, $n_{th} = (e^{\frac{\hbar\omega}{k_{B}T}}-1)^{-1}$, in the input reservoir field to the out-coupling port. 
For an optical system with $\omega/(2\pi) =$ 300 THz at room temperature $T = $300 K, the thermal photon number is $n_{th} = 0.02$, and the result is indistinguishable from the ideal DOPO system at absolute zero temperature, i.e., for $n_{th} = 0$ shown in Fig.12. 
However, if the thermal photon number $n_{th} > 1$, the thermal noise effect becomes apparent. 
%
%
%
Figure 14 shows the negative correlation ${\rm C}(\hat{x}_{s1}\hat{x}_{s2})$ defined by (14) vs.~normalized pump rate $p$ for a different thermal photon number $n_{th}$.
The correlation is degraded and eventually vanish as the thermal photon number increases, which is responsible for the decreased success probability with increasing $n_{th}$ shown in Fig.~\ref{fig:ans_TP}.
From those results shown in Fig.~\ref{fig:ans_TP} and Fig.~\ref{fig:cor_TP}, we can conclude the quantum oscillator network operating at standard vacuum fluctuation limit or squeezed vacuum fluctuation limit rather than the classical oscillator network operating at thermal noise limit is the key to the successful performance of the CIM.
\newpage
\section{Conclusion}
We studied the quantum correlation, inseparability and probability of finding the ground state in a one-dimensional ring configuration consisting of identical anti-ferromagnetically coupled Ising spins.
The validation of  the theoretical method based on the truncated Wigner distribution function was checked by comparing the computed inseparability with that obtained by a more rigorous method based on the positive-$P$ representation (off-diagonal coherent state expansion). 
The success probability drops dramatically when the system is subject to large thermal noise, which indicates that the quantum parallel search, quantum filtering, spontaneous symmetry breaking, and quantum-to-classical crossover transiently realized in the network of quantum oscillators play a crucial role in the CIM. 
We also demonstrated that the non-classical read out of the in-phase amplitude $c_{si,cav}$ of the signal field with a squeezed vacuum state input increases the probability of finding the ground state in the case of 1D Ising spins with anti-ferromagnetic coupling. 
Implementation of such a non-classical read out only requires another DOPA with an appropriate pump phase; hence, the proposed scheme can be realized without having to deal with any serious technical challenges.
\clearpage
\begin{figure}[!hbp]
\centering
\includegraphics[clip=true,scale=0.6]{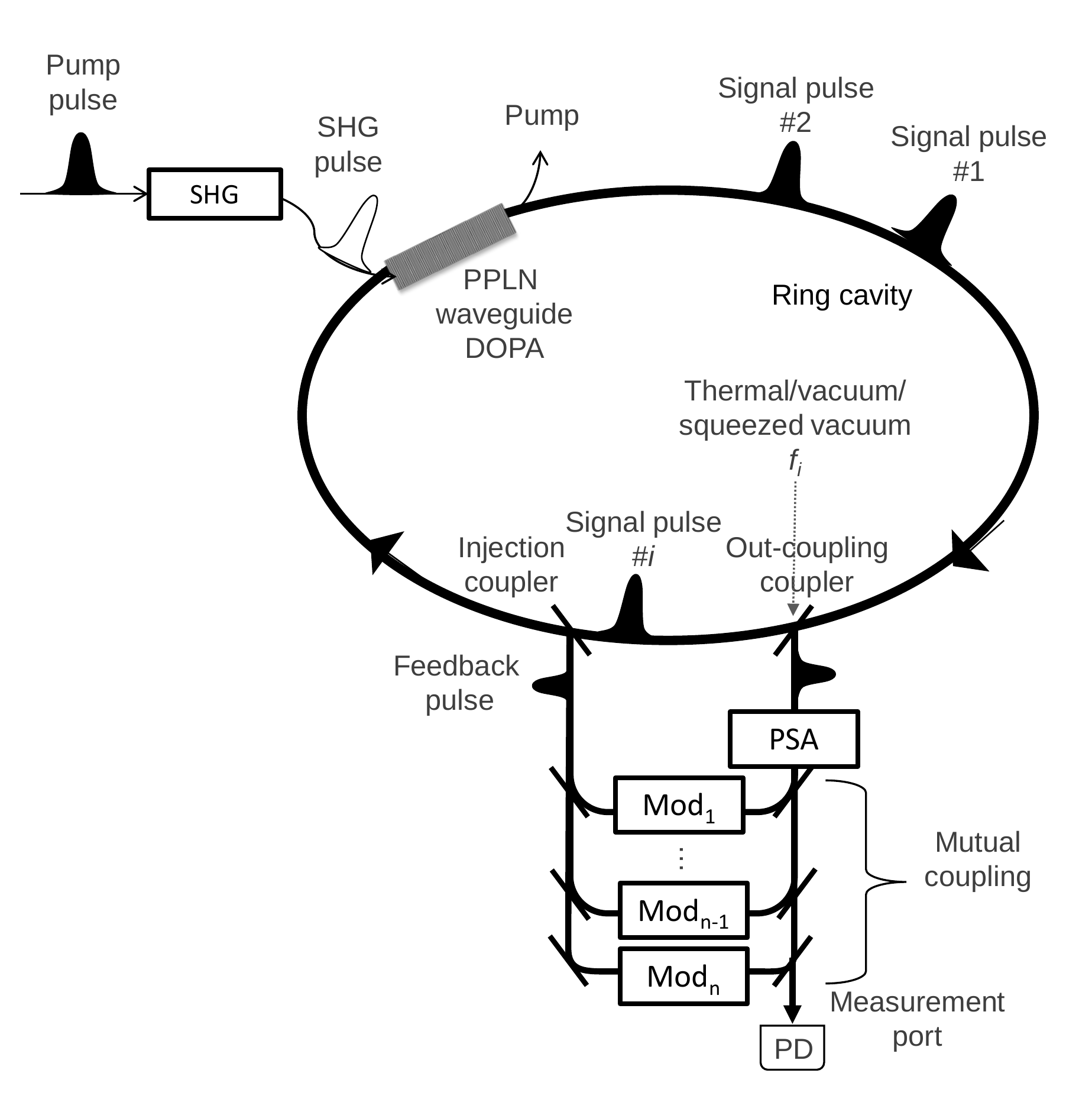}
\bigskip
\caption{A coherent Ising machine with optical delay line coupling. The output coupler followed by the phase sensitive amplifier (PSA: degenerate parametric amplifier) amplifies the in-phase amplitude $\hat{x}$ of each DOPO pulse, while the injection coupler combines the modulated feedback pulse with the target DOPO pulse, which implements the given Ising Hamiltonian. 
The state incident to the output coupler from an open port plays an important role in the behavior of this system and the final success probability of CIM.}
\label{fig:CIM}
\end{figure}

\newpage
\begin{figure}[!hbp]
\centering
\includegraphics[clip=true,scale=0.6]{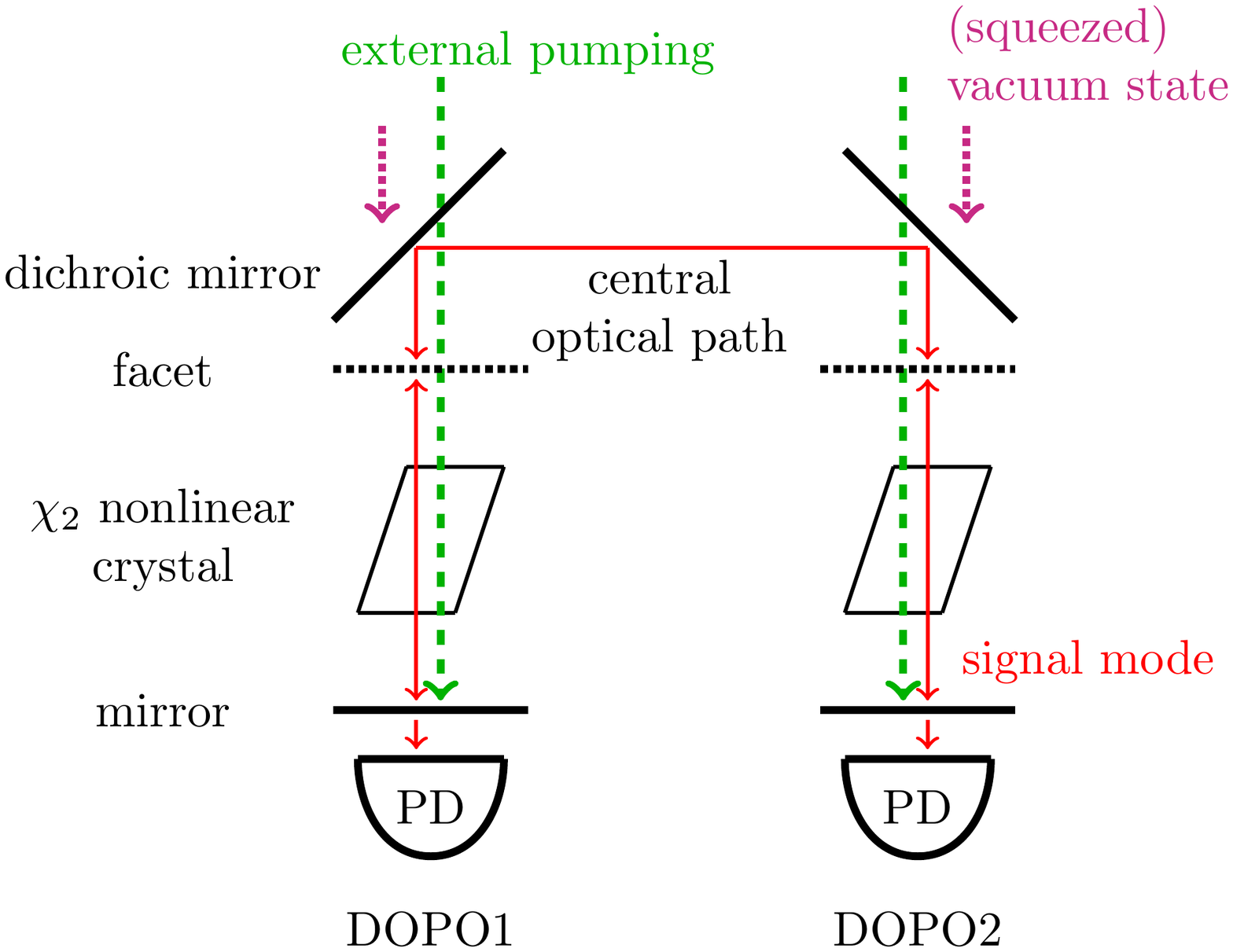}
\caption{Sketch of two DOPOs coupled via a central optical path. Each DOPO consists of the nonlinear crystal (parallelogram) and two mirrors (bold and dotted horizontal lines). The upper ones are the facets. The central optical path is the space between the facets and dichroic mirrors (tilted bold line). Red bold arrows express the signal modes (including the center mode), green dashed arrows express external pumping, and magenta dotted arrows express the incident (squeezed) vacuum state into the central optical path.}
\label{fig:2DOPO}
\end{figure}

\clearpage
\begin{figure}[!hbp]
\centering
\includegraphics[clip=true,scale=0.5]{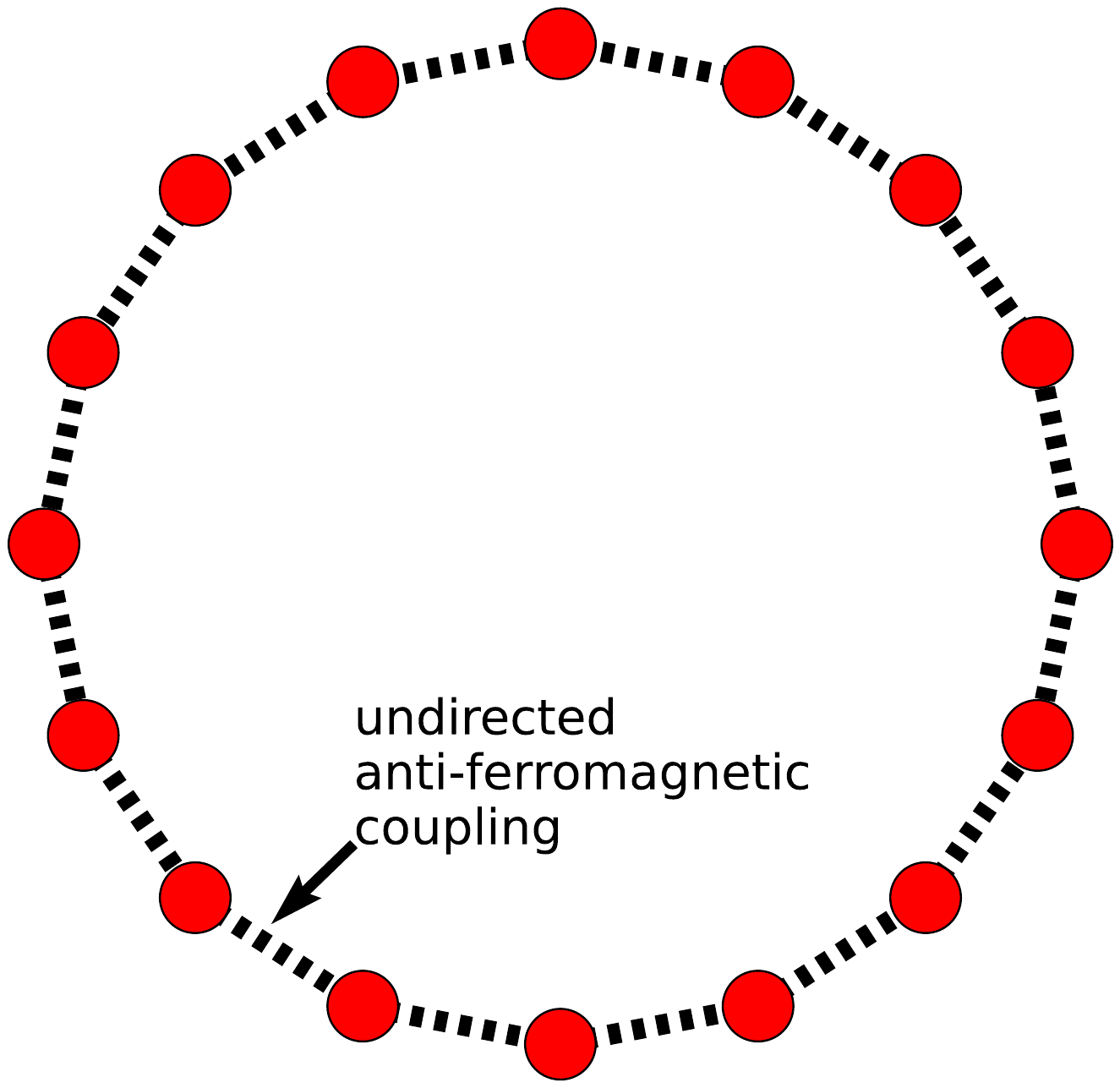}
\bigskip
\caption{Sketch of one-dimensional ring consisting of 16 DOPOs in which 16 Ising spins are coupled with nearest-neighbor identical anti-ferromagnetic couplings.}
\label{fig:ring}
\end{figure}

\clearpage
\begin{figure}[!hbp]
\centering
\includegraphics[clip=true,scale=0.5]{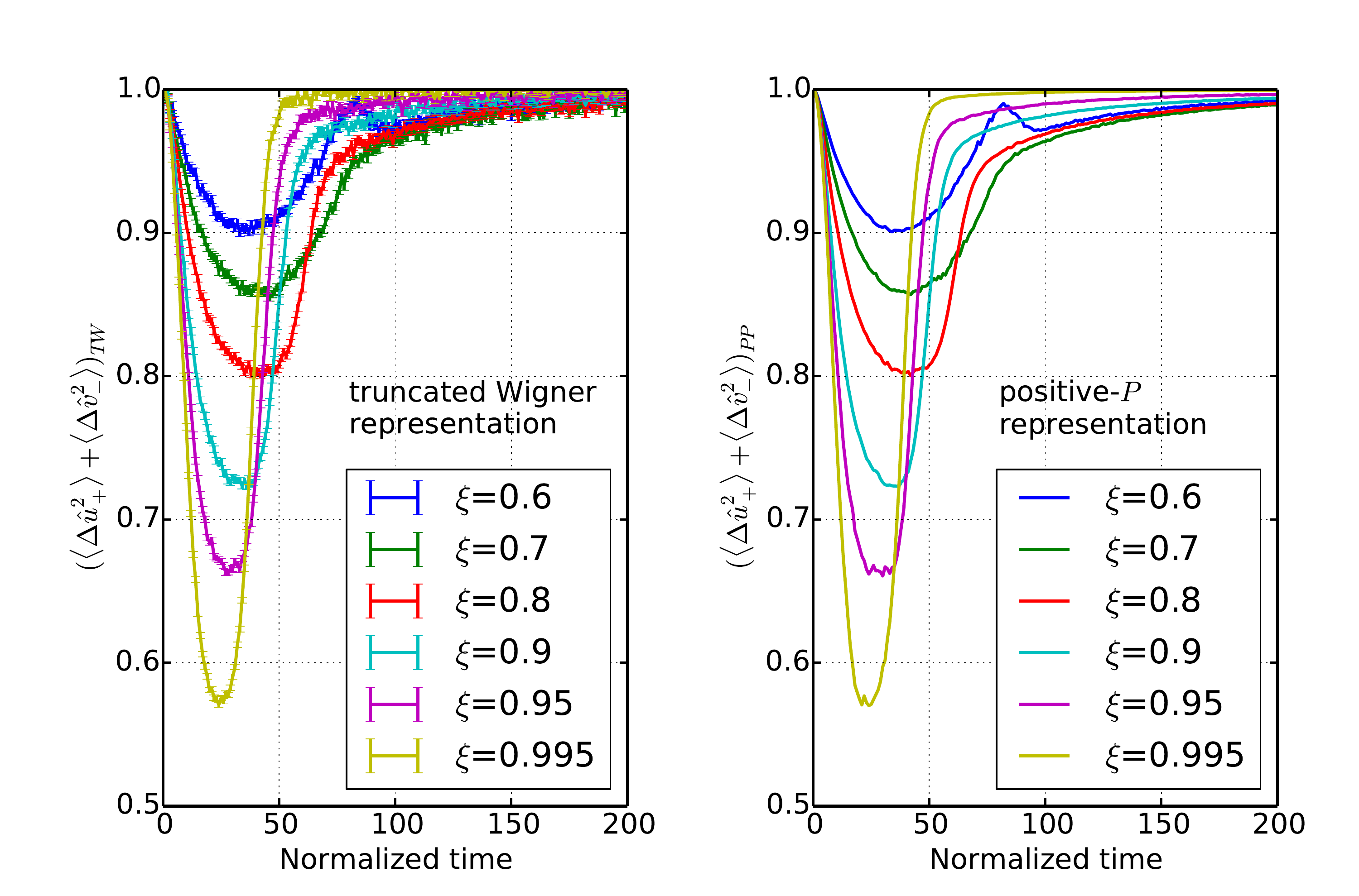}
\bigskip
\caption{Total variance of the EPR operator $\hat{u}_{+} + \hat{v}_{-}$ calculated by the truncated Wigner representation (left panel,(a)) and positive-$P$ representation (right panel,(b)).The statistical error bars due to the finite number of samples $N=200.000$ are only plotted in Fig.~\ref{fig:ent}(a).}
\label{fig:ent}
\end{figure}

\clearpage
\begin{figure}[!hbp]
\centering
\includegraphics[clip=true,scale=0.5]{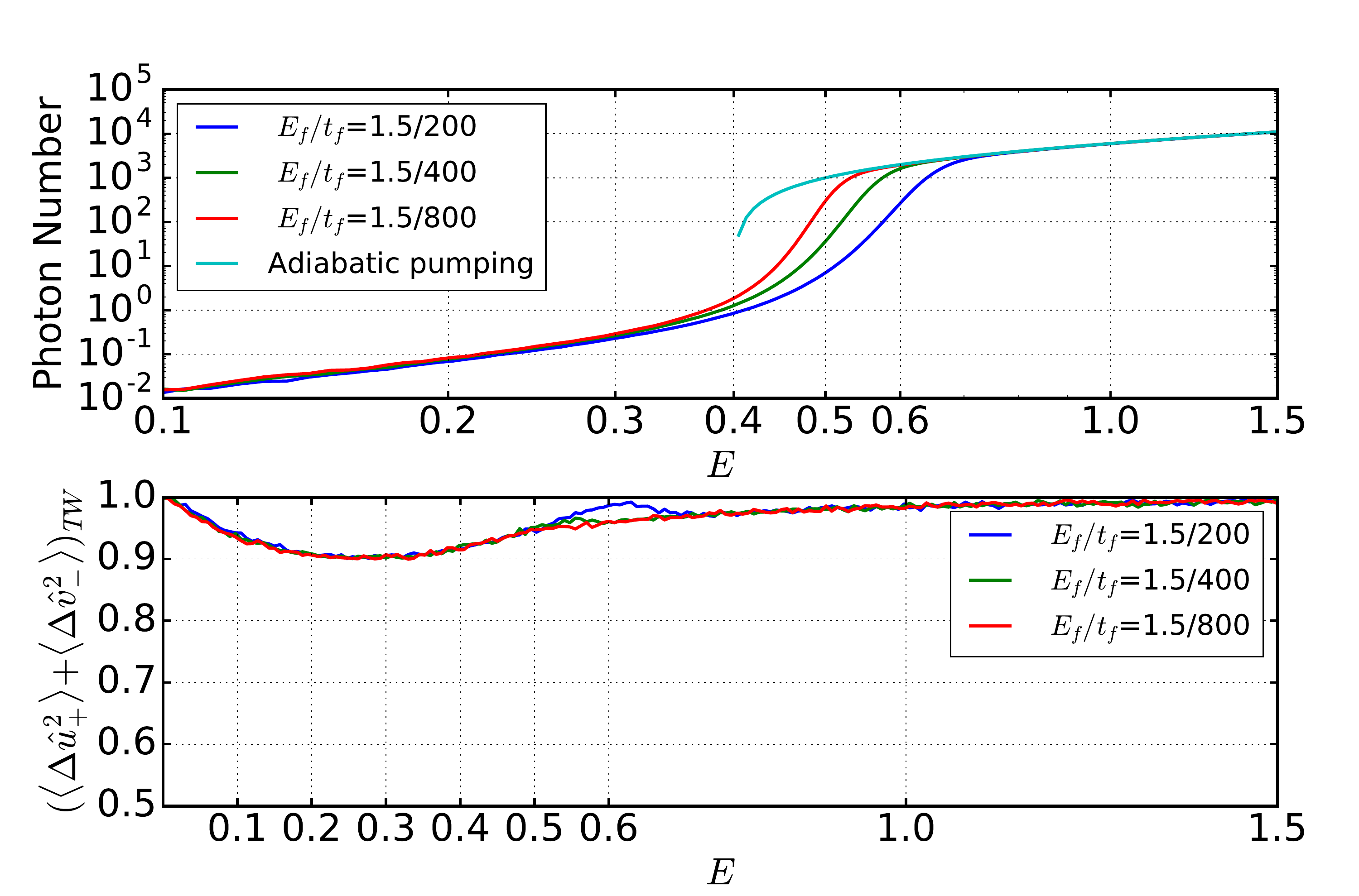}
\bigskip
\caption{The average photon number $\average{n}$ and the variance of the EPR operator $\hat{u}_{+} + \hat{v}_{-}$ vs. normalized pumping rate $E$ for varied pump schedule.}
\label{fig:noneq}
\end{figure}

\clearpage
\begin{figure}[!hbp]
\centering
\includegraphics[clip=true,scale=0.5]{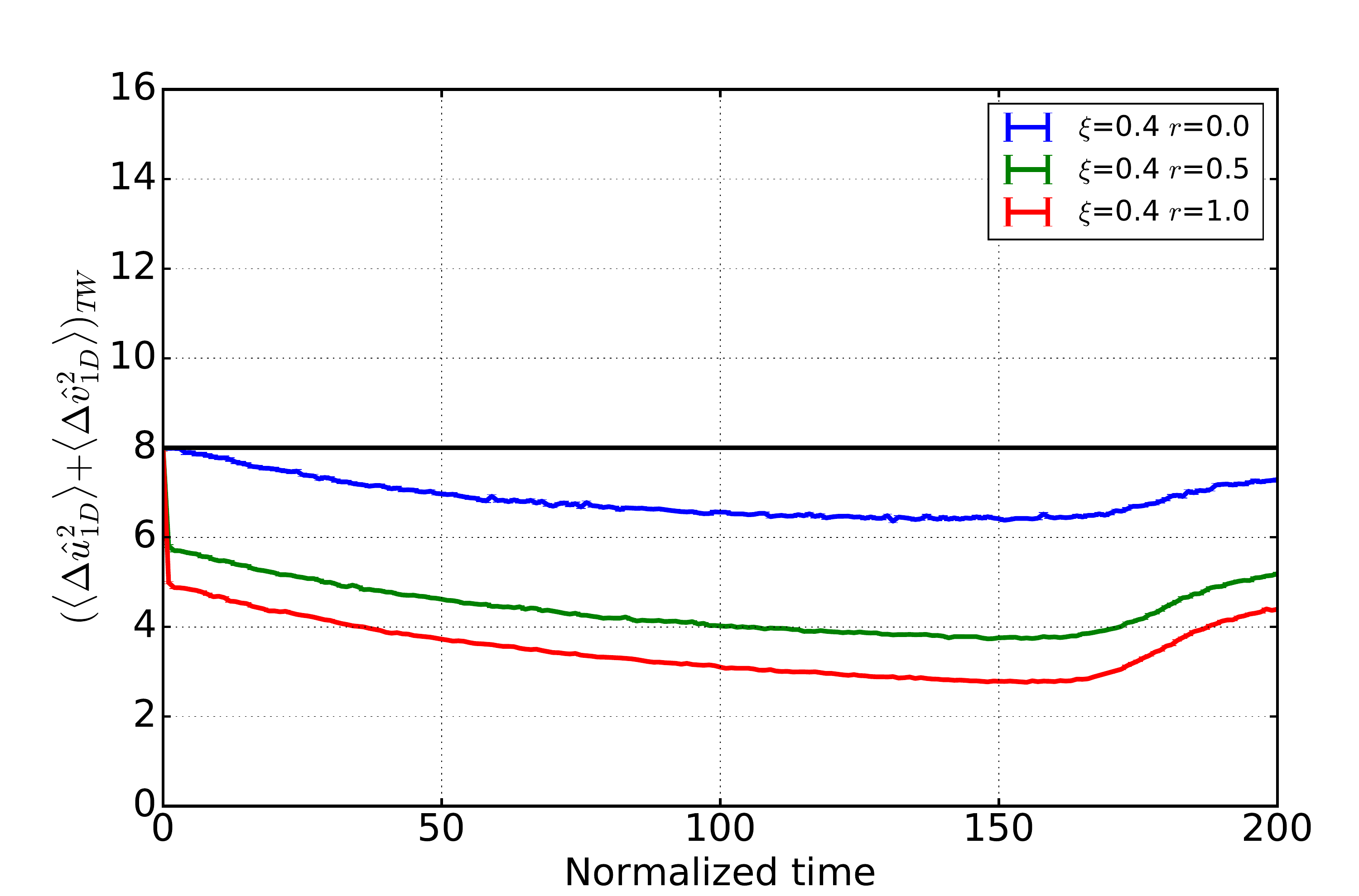}
\bigskip
\caption{The variance of $\hat{u}_{1D} + \hat{v}_{1D}$ vs. normalized time $\tau$ for various squeezing parameters.}
\label{fig:insep}
\end{figure}


\clearpage
\begin{figure}[!hbp]
\centering
\begin{subfigure}{0.48 \textwidth}
\includegraphics[width=\linewidth]{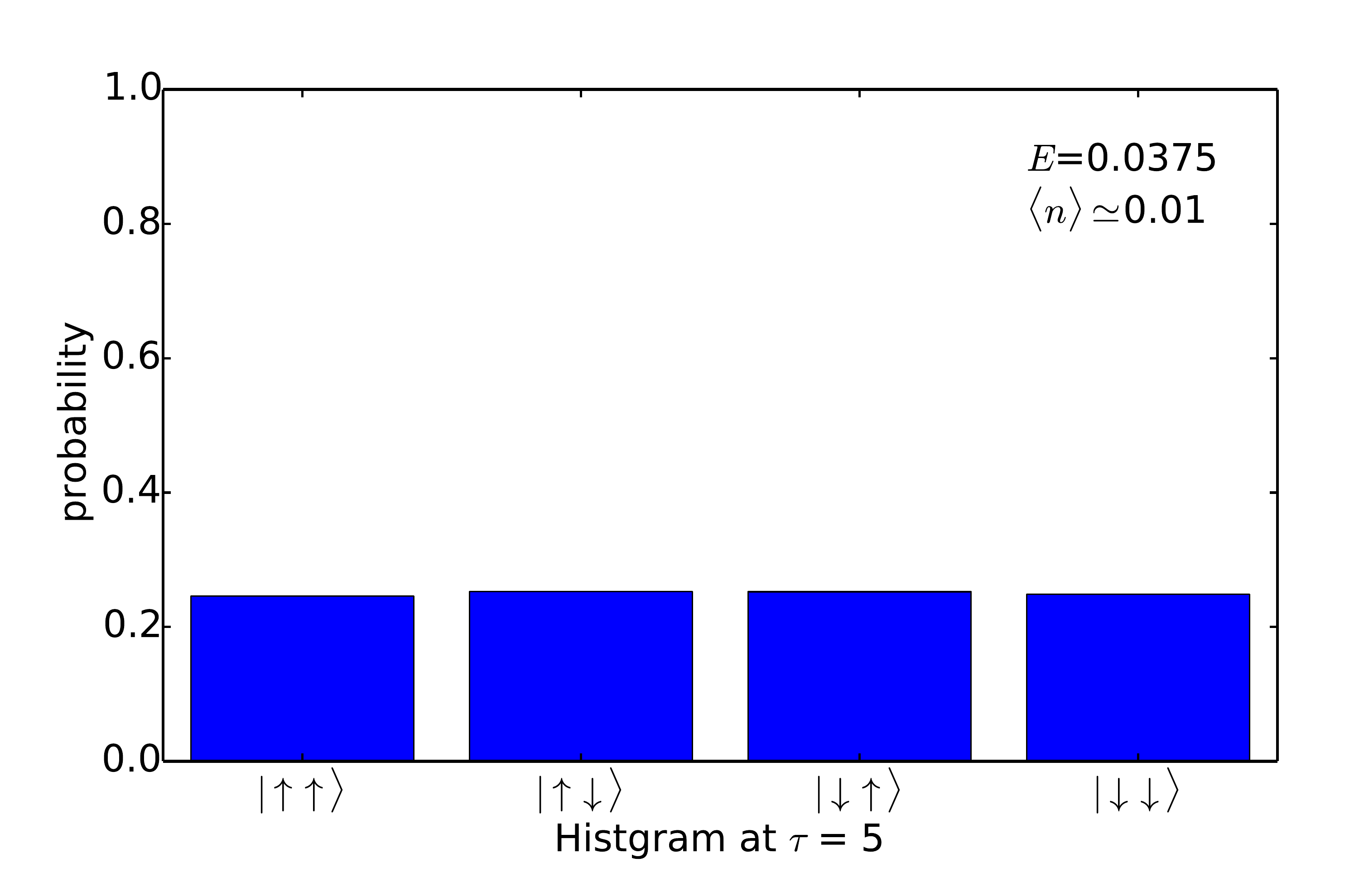}
\caption{} \label{fig:tau_5}
\end{subfigure}
\begin{subfigure}{0.48 \textwidth}
\includegraphics[width=\linewidth]{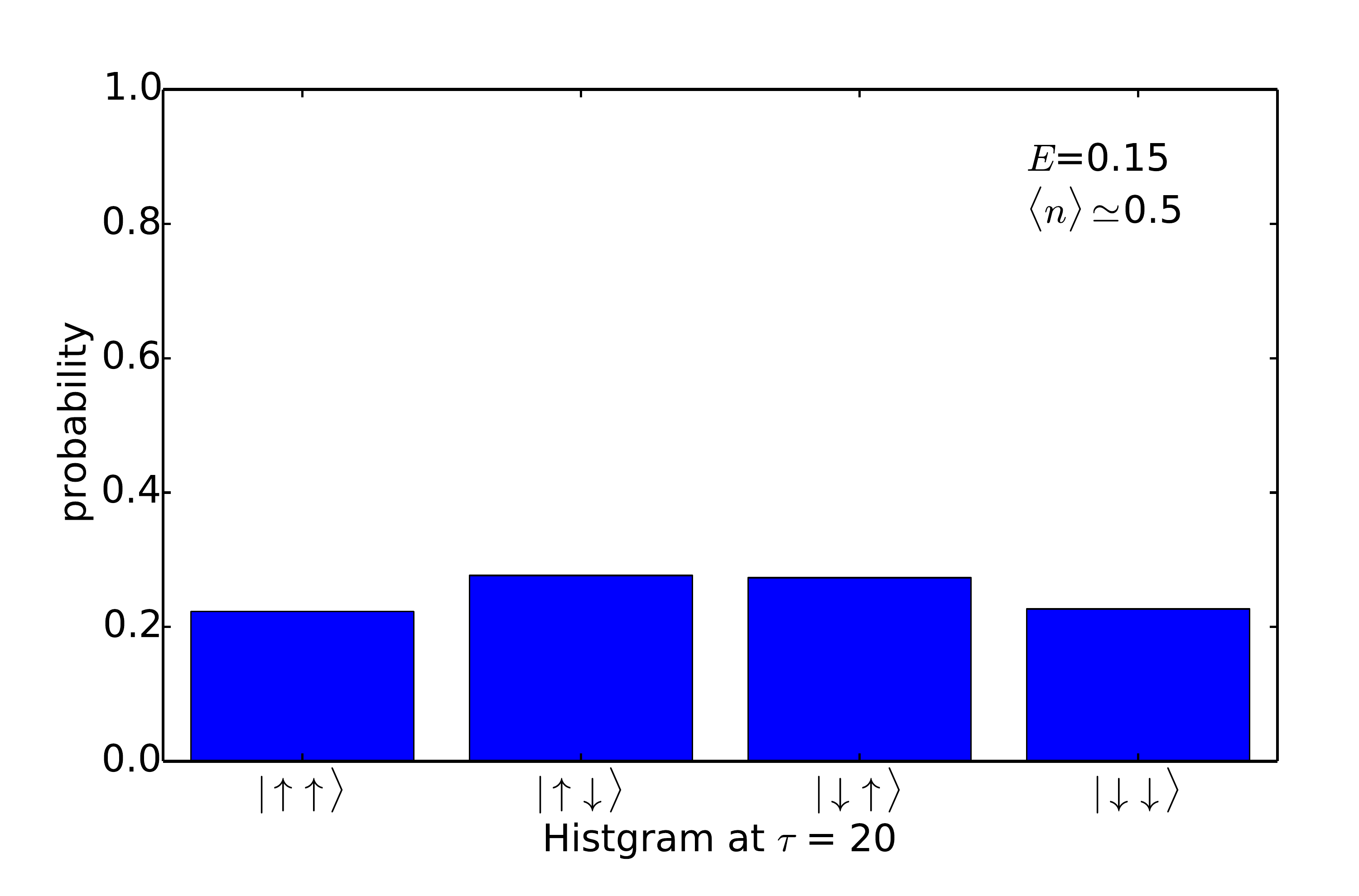}
\caption{} \label{fig:tau_20}
\end{subfigure}
\begin{subfigure}{0.48 \textwidth}
\includegraphics[width=\linewidth]{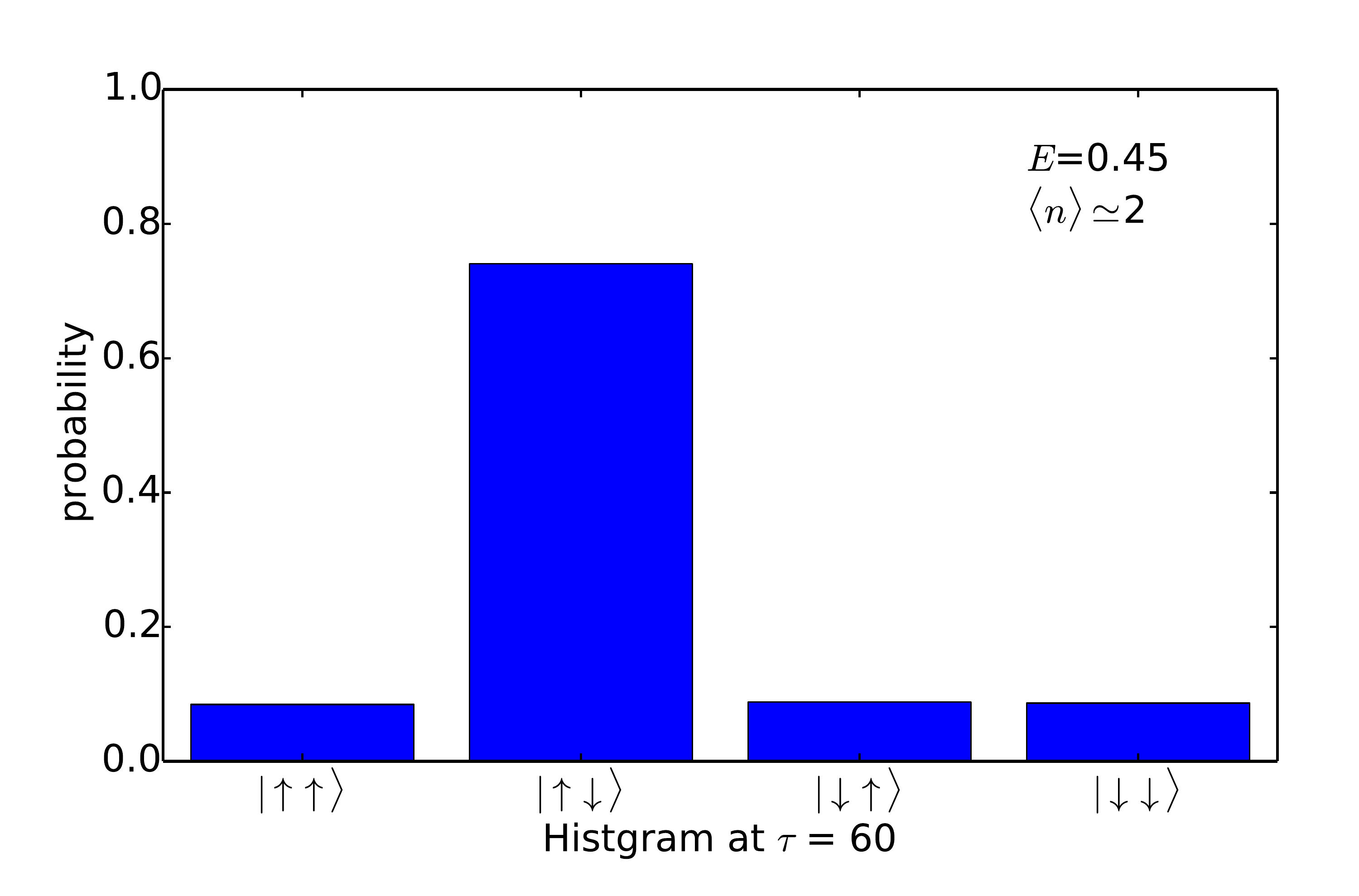}
\caption{} \label{fig:tau_60}
\end{subfigure}
\begin{subfigure}{0.48 \textwidth}
\includegraphics[width=\linewidth]{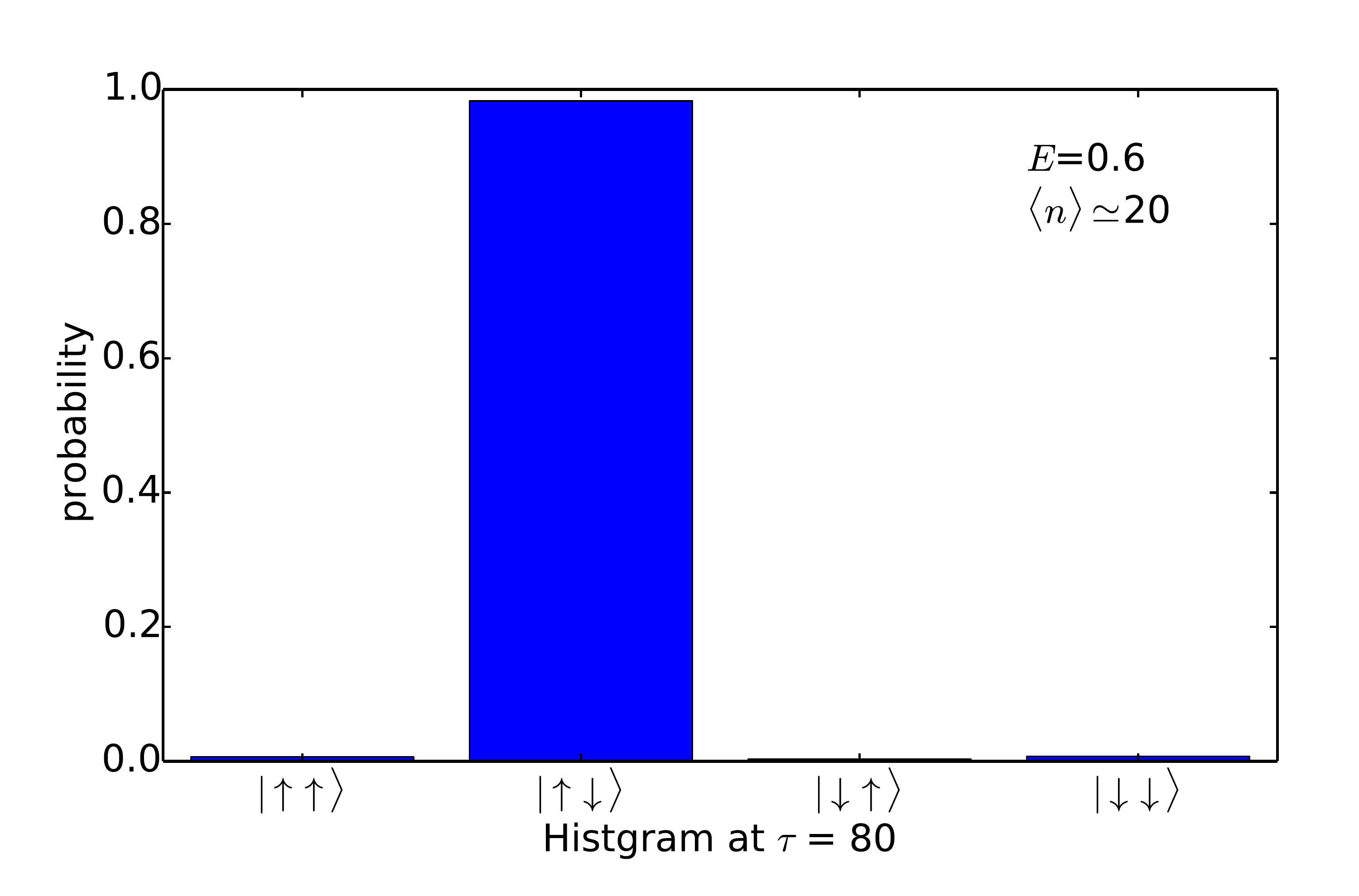}
\caption{} \label{fig:tau_80}
\end{subfigure}
\caption{Probabilities of finding $\ket{\uparrow \uparrow}, \ket{\uparrow \downarrow}, \ket{\downarrow \uparrow}$ and $\ket{\downarrow \downarrow}$ states at four different times $\tau$, pump rates $E$, and average photon numbers $\average{n}$ per DOPO when the final result is $\ket{\uparrow \downarrow}$. The four panels correspond to the computational stages of quantum parallel search, quantum filtering, spontaneous symmetry breaking, and quantum-to-classical crossover.} 
\label{fig:bar}
\end{figure}

\clearpage
\begin{figure}[!hbp]
\centering
\begin{subfigure}{0.8 \textwidth}
\includegraphics[width=\linewidth]{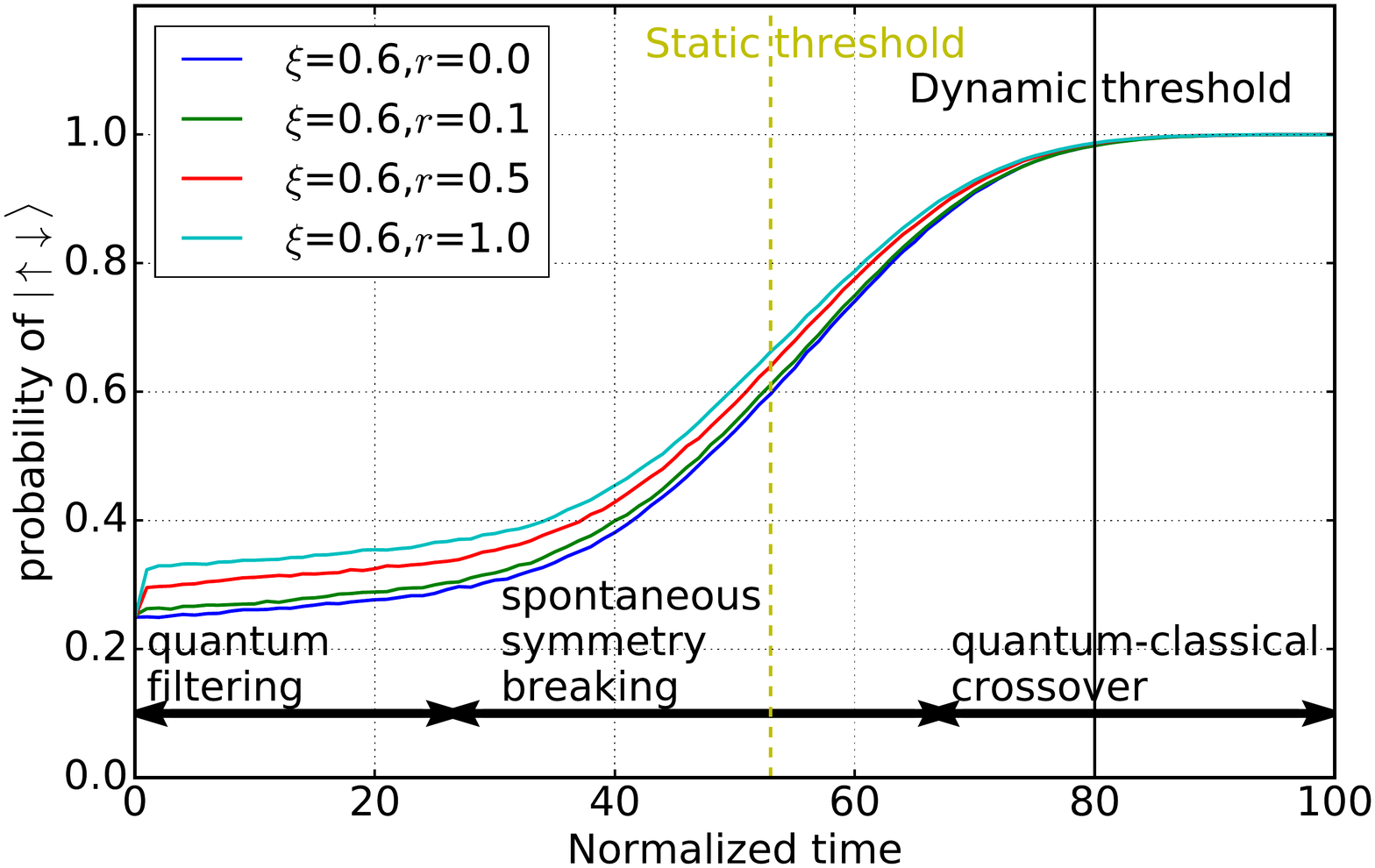}
\caption{} \label{fig:udud}
\end{subfigure}
\begin{subfigure}{0.8 \textwidth}
\includegraphics[width=\linewidth]{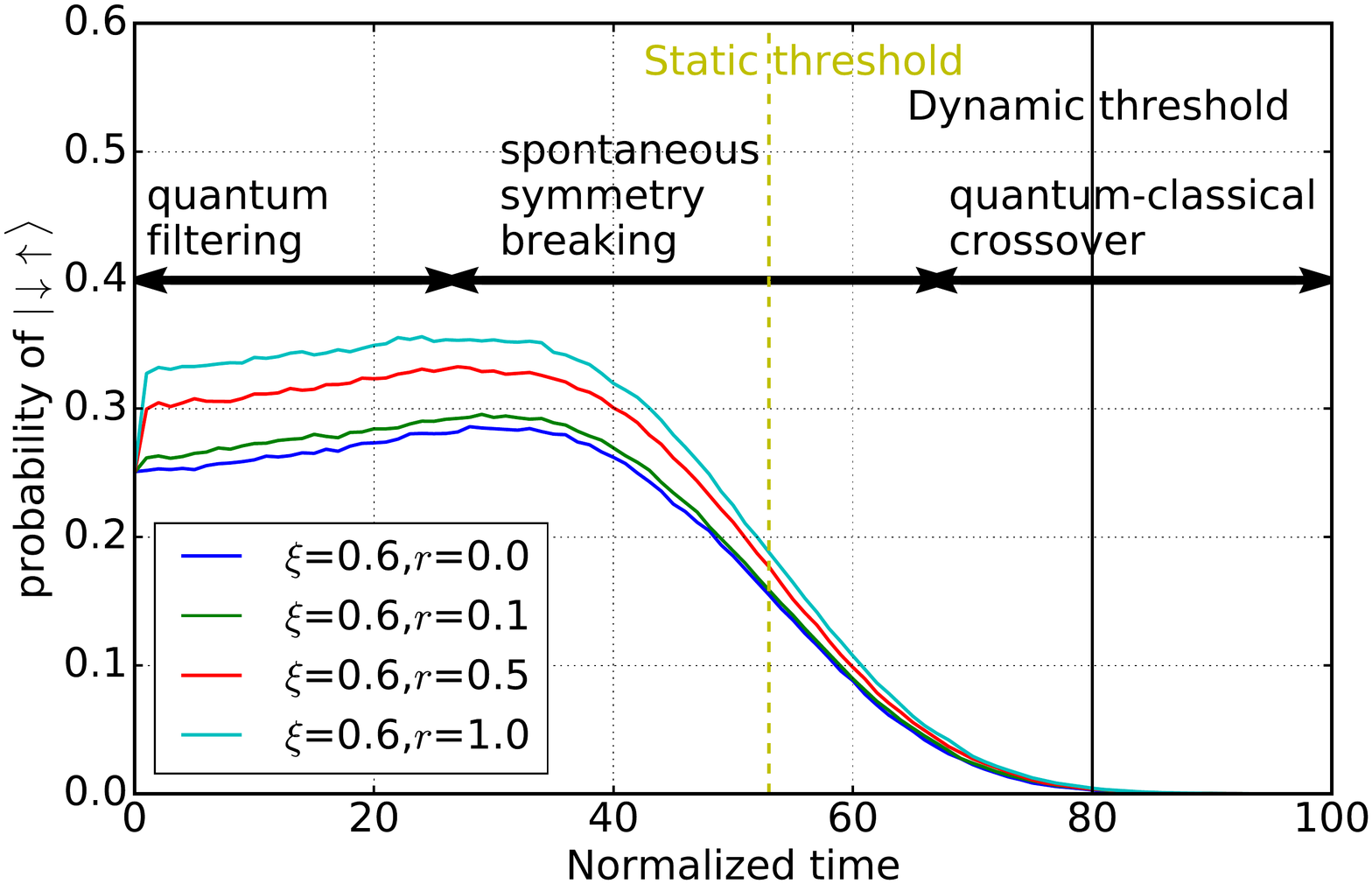}
\caption{} \label{fig:uddu}
\end{subfigure}
\caption{Probabilities of finding the selected state $\ket{\uparrow \downarrow}$ and the unselected state $\ket{\downarrow \uparrow}$ vs. normalized time.  $r$ is the squeezing parameter and the variance of in-phase amplitude noise incident on the output coupler is given by $(1/4)e^{-2r}$, where $r=0$ corresponds to the standard vacuum state. The system reaches the oscillation threshold $E_{th} = 1- \xi = 0.4$ (static threshold) at the normalized time $\tau \simeq 53$. However, due to the turn-on delay effect, the actual oscillation occurs at $\tau \simeq 80$ (dynamic threshold).}
\label{fig:post}
\end{figure}

\clearpage
\begin{figure}[!hbp]
\centering
\includegraphics[clip=true,scale=0.5]{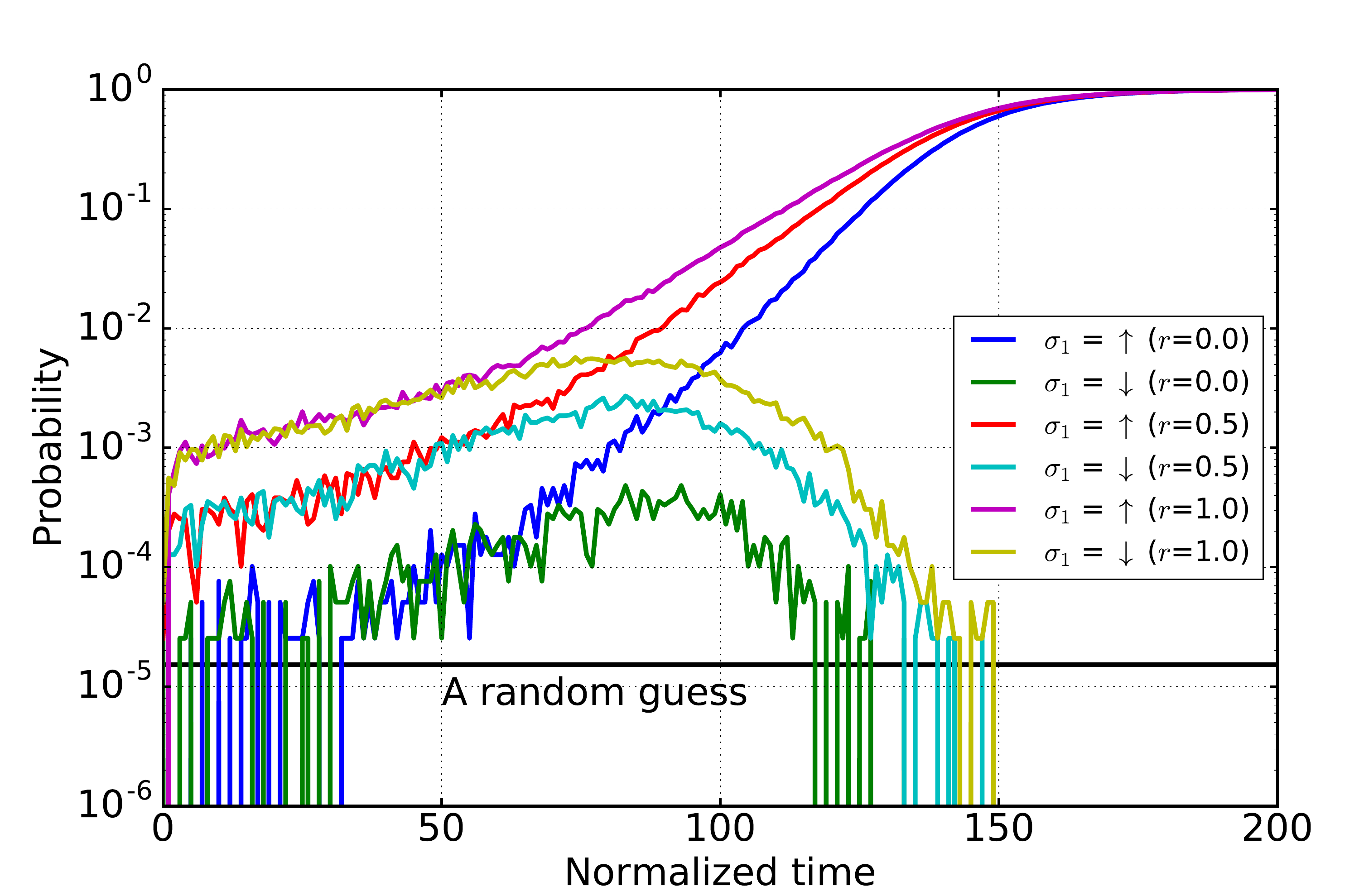}
\bigskip
\caption{Probabilities of finding the selected ground state$\ket{\uparrow \downarrow \uparrow \downarrow \uparrow \downarrow \uparrow \downarrow \uparrow \downarrow \uparrow \downarrow \uparrow \downarrow \uparrow \downarrow}$($\sigma_1 = \uparrow$) and the unselected ground state $\ket{\downarrow \uparrow \downarrow \uparrow \downarrow \uparrow \downarrow \uparrow \downarrow \uparrow \downarrow \uparrow \downarrow \uparrow \downarrow \uparrow}$($\sigma_1 = \downarrow$) for various squeezing parameter $r$. The system performs quantum filtering at $5<\tau<50$ and spontaneous symmetry breaking at $60<\tau<130$ after the very brief period of quantum parallel search at $0<\tau<5$. Finally the amplitude of electromagnetic field become large enough to measure the spins (quantum-classical crossover).}
\label{fig:searching}
\end{figure}

\clearpage
\begin{figure}[!hbp]
\centering
\includegraphics[clip=true,scale=0.5]{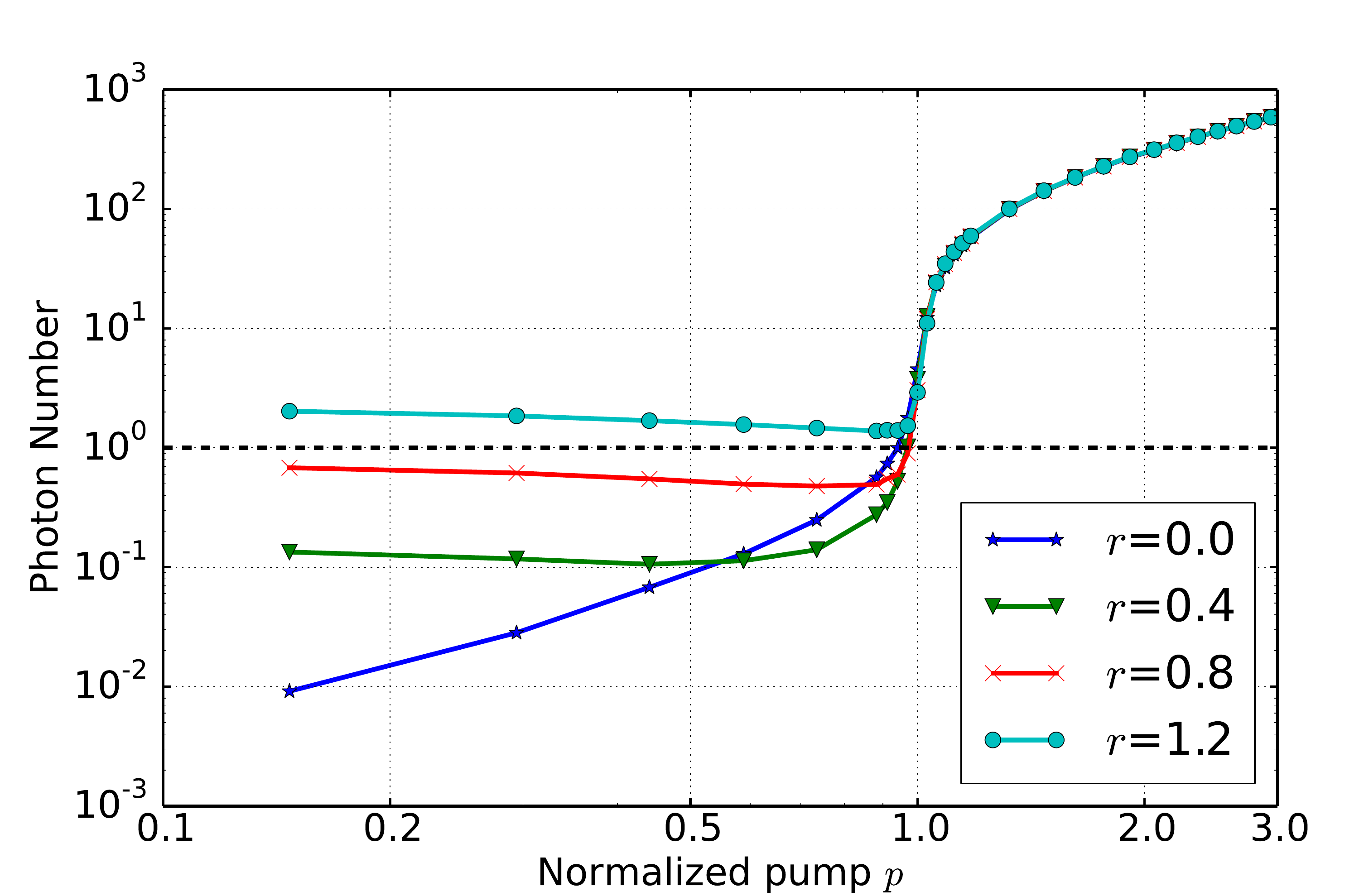}
\bigskip
\caption{Average photon number $\average{n}$ per DOPO at 2000 round trips after the pump is switched on vs.~normalized pump rate $p$ for different squeezing parameters $r$. }
\label{fig:pn}
\end{figure}

\clearpage
\begin{figure}[!hbp]
\centering
\begin{subfigure}{0.7 \textwidth}
\includegraphics[width=\linewidth]{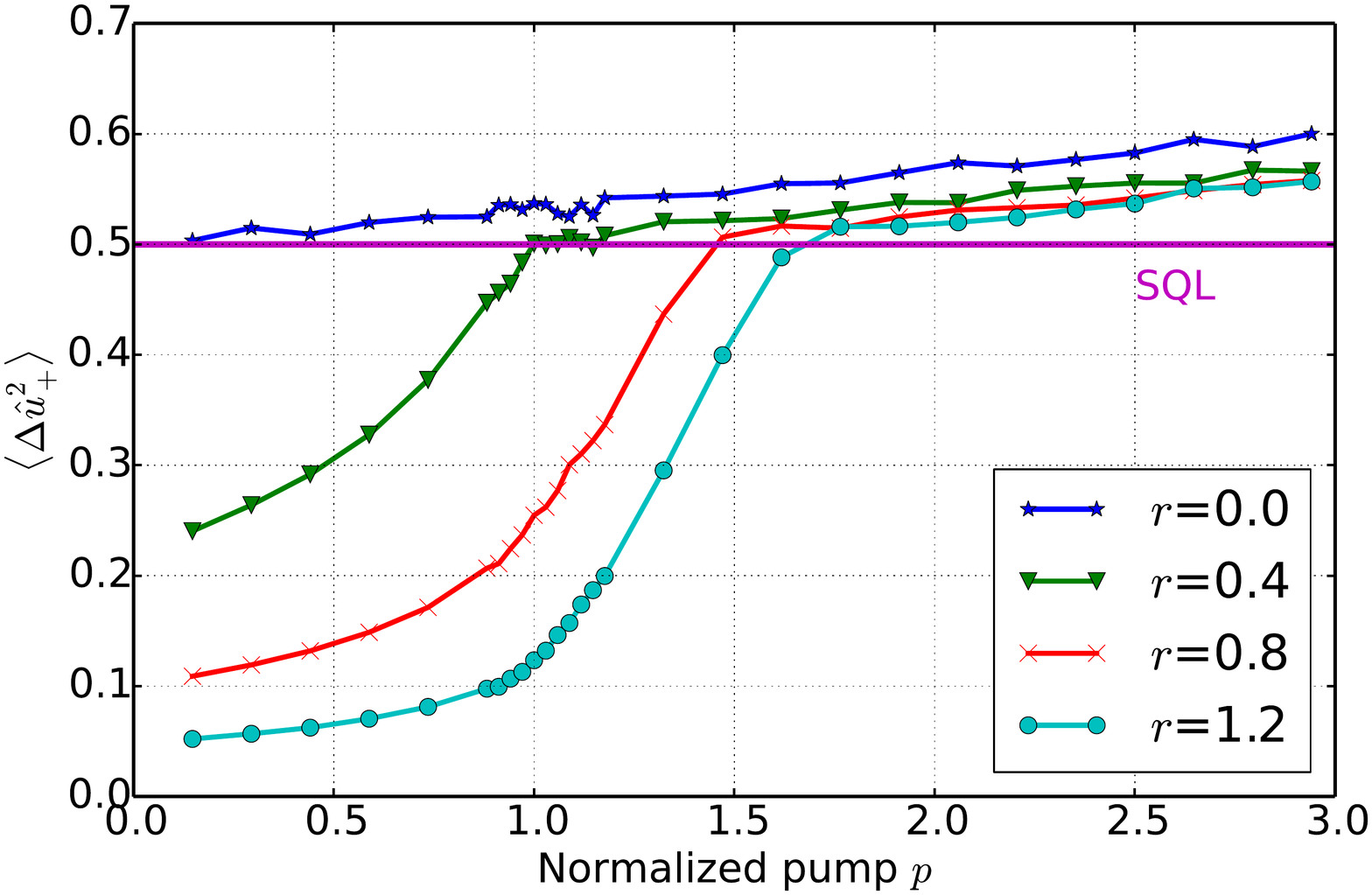}
\caption{}
\label{fig:dudu_min}
\end{subfigure}
\begin{subfigure}{0.7 \textwidth}
\includegraphics[width=\linewidth]{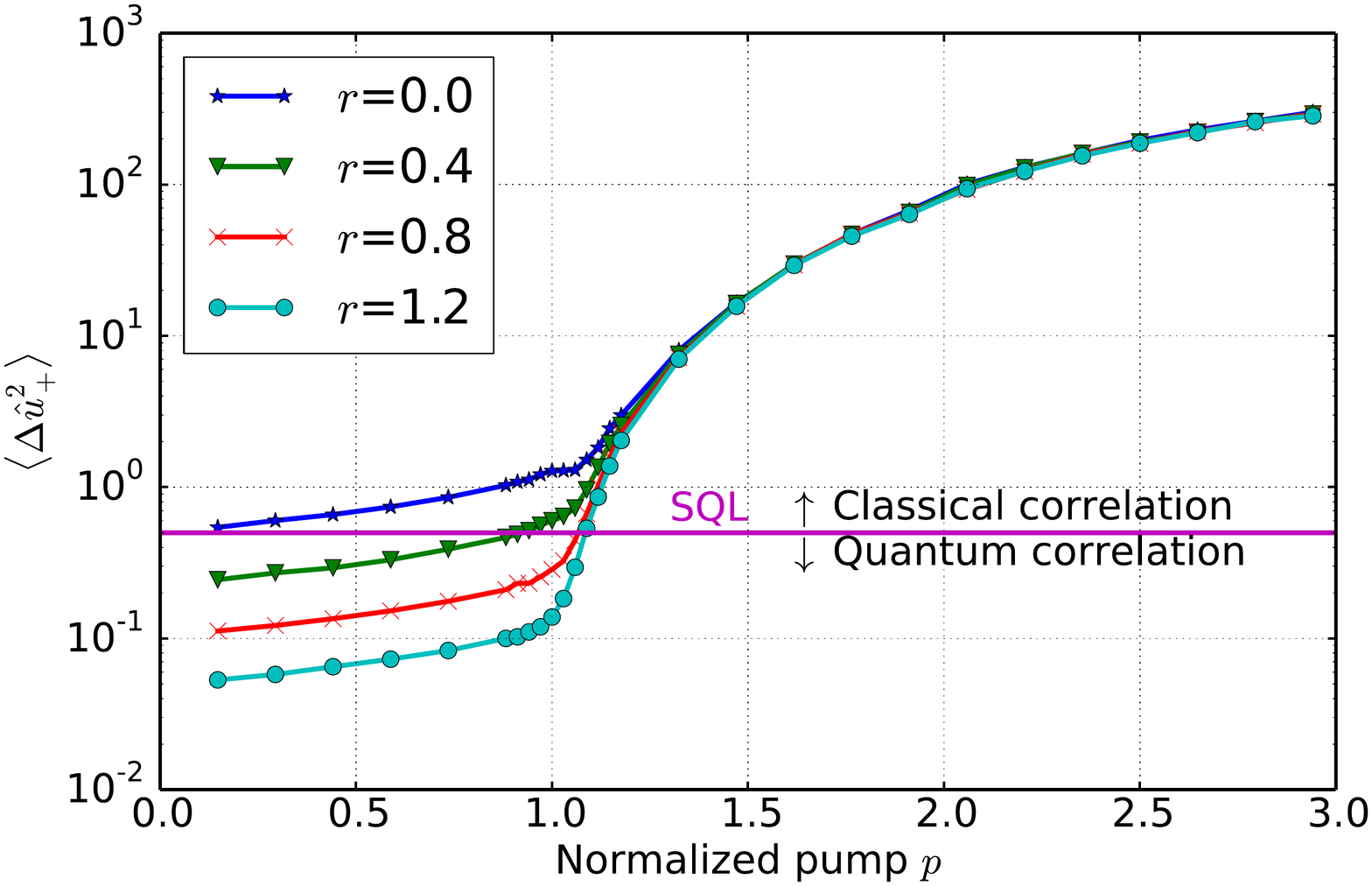}
\caption{}
\label{fig:dudu_fin}
\end{subfigure}
\caption{Variance of $\hat{u}_{+}$ between two neighboring spins vs.~normalized pump rate $p$ for different squeezing parameters  $r$: \\ (a) Minimum variance of $\hat{u}_{+}$ during 2000 round trips vs.~normalized pump rate $p$. \\ (b) Final variance of $\hat{u}_{+}$ after 2000 round trips vs.~normalized pump rate $p$.} 
\label{fig:dudu}
\end{figure}

\clearpage
\begin{figure}[!hbp]
\centering
\includegraphics[clip=true,scale=0.5]{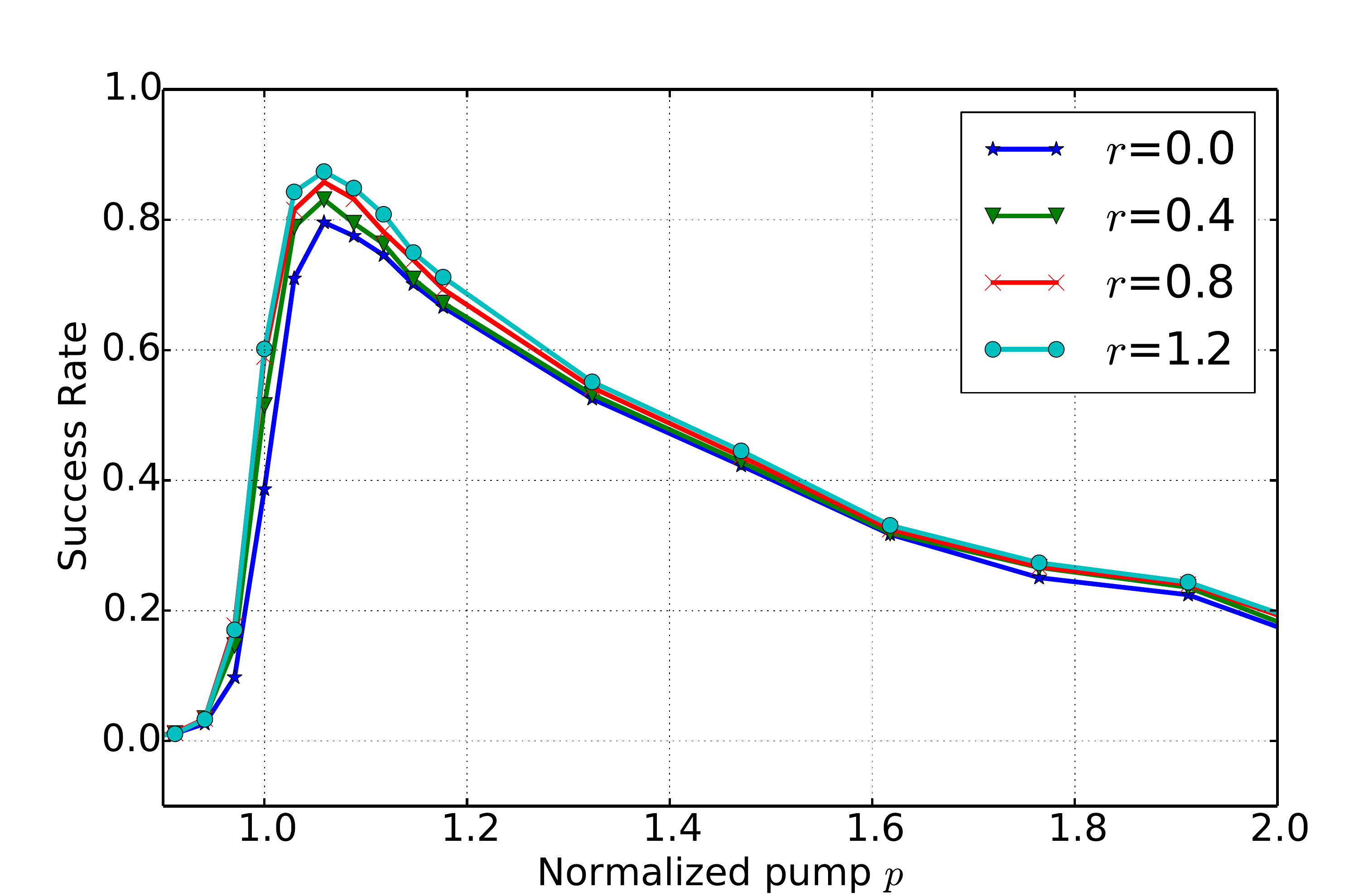}
\bigskip
\caption{Probability of finding a ground state after 2000 round trips vs.~normalized pump rate $p$ for different squeezing parameters  $r$.}
\label{fig:ans}
\end{figure}

\clearpage
\begin{figure}[!hbp]
\centering
\includegraphics[clip=true,scale=0.5]{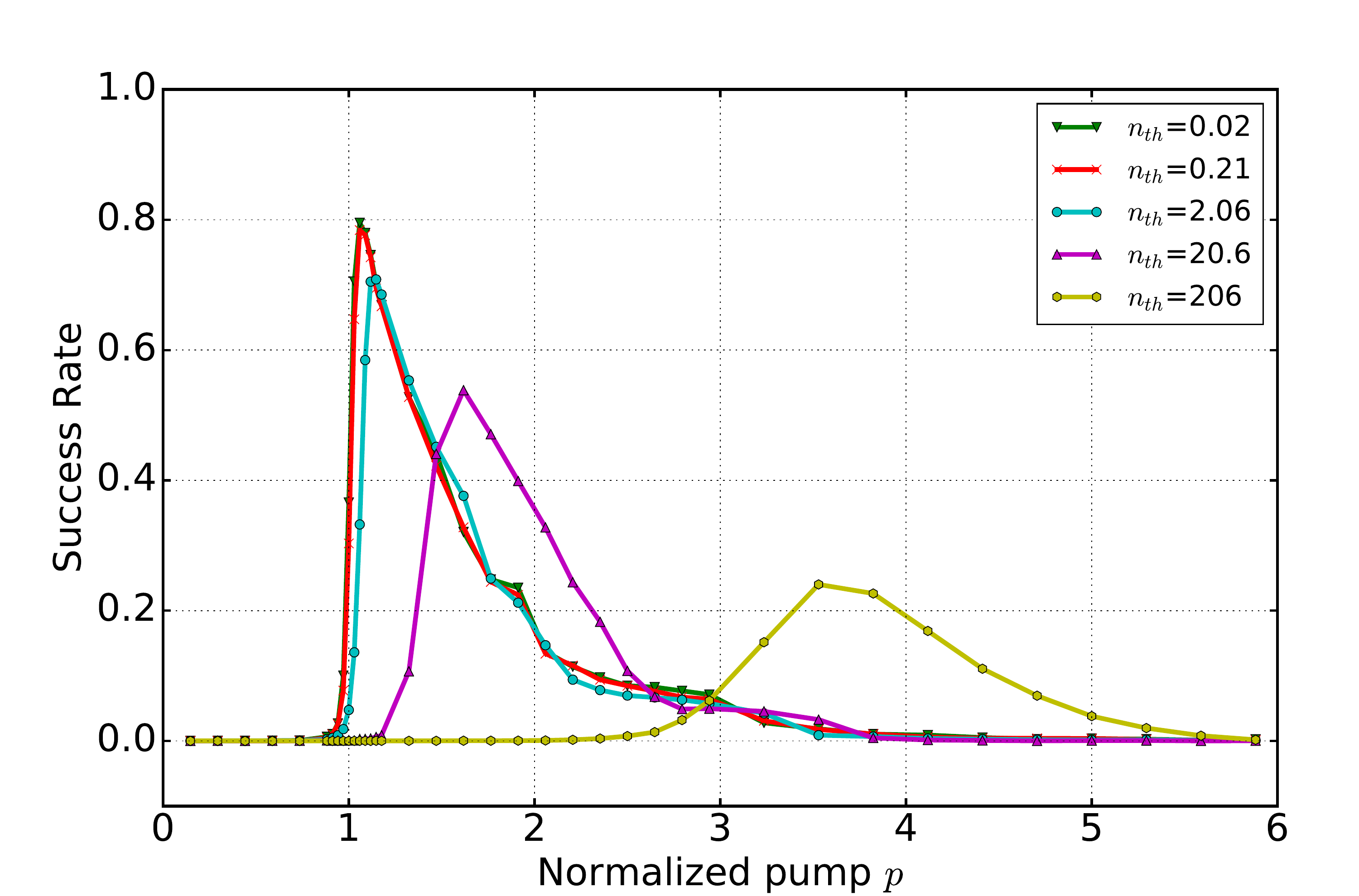}
\bigskip
\caption{Probability of finding a ground state after 2000 round trips vs.~the normalized pump rate $p$ for different thermal photon numbers  $n_{th}$.}
\label{fig:ans_TP}
\end{figure}

\clearpage
\begin{figure}[!hbp]
\centering
\includegraphics[clip=true,scale=0.5]{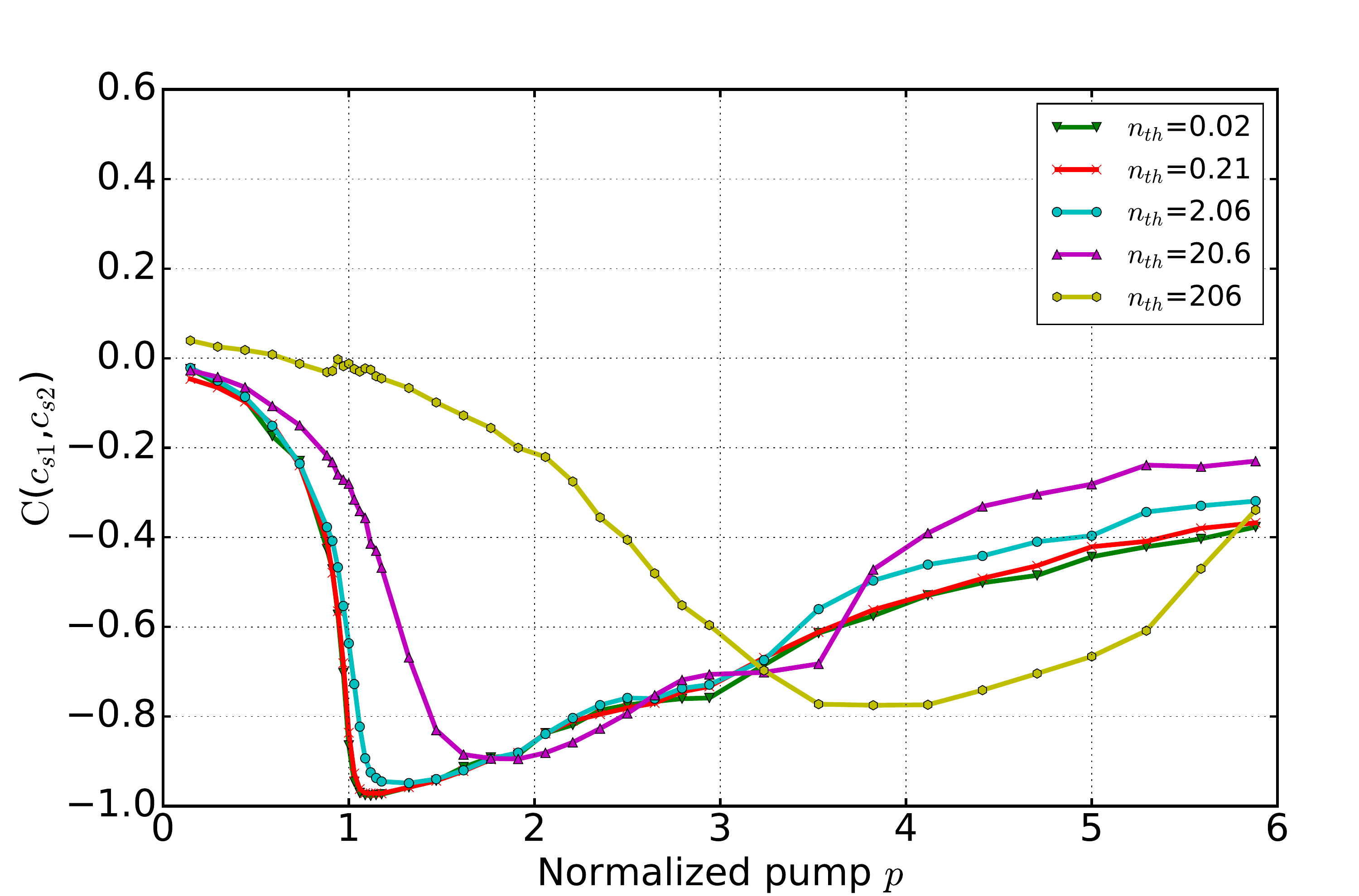}
\bigskip
\caption{The normalized correlation ${\rm C}(\hat{x}_{s1},\hat{x}_{s2})$ after 2000 round trips vs.~the normalized pump rate $p$ for different thermal photon numbers  $n_{th}$.}
\label{fig:cor_TP}
\end{figure}

\clearpage
\bibliographystyle{apsrev4-1}
\bibliography{submission.bib}

\begin{thebibliography}{26}%
\makeatletter
\providecommand \@ifxundefined [1]{%
 \@ifx{#1\undefined}
}%
\providecommand \@ifnum [1]{%
 \ifnum #1\expandafter \@firstoftwo
 \else \expandafter \@secondoftwo
 \fi
}%
\providecommand \@ifx [1]{%
 \ifx #1\expandafter \@firstoftwo
 \else \expandafter \@secondoftwo
 \fi
}%
\providecommand \natexlab [1]{#1}%
\providecommand \enquote  [1]{``#1''}%
\providecommand \bibnamefont  [1]{#1}%
\providecommand \bibfnamefont [1]{#1}%
\providecommand \citenamefont [1]{#1}%
\providecommand \href@noop [0]{\@secondoftwo}%
\providecommand \href [0]{\begingroup \@sanitize@url \@href}%
\providecommand \@href[1]{\@@startlink{#1}\@@href}%
\providecommand \@@href[1]{\endgroup#1\@@endlink}%
\providecommand \@sanitize@url [0]{\catcode `\\12\catcode `\$12\catcode
  `\&12\catcode `\#12\catcode `\^12\catcode `\_12\catcode `\%12\relax}%
\providecommand \@@startlink[1]{}%
\providecommand \@@endlink[0]{}%
\providecommand \url  [0]{\begingroup\@sanitize@url \@url }%
\providecommand \@url [1]{\endgroup\@href {#1}{\urlprefix }}%
\providecommand \urlprefix  [0]{URL }%
\providecommand \Eprint [0]{\href }%
\providecommand \doibase [0]{http://dx.doi.org/}%
\providecommand \selectlanguage [0]{\@gobble}%
\providecommand \bibinfo  [0]{\@secondoftwo}%
\providecommand \bibfield  [0]{\@secondoftwo}%
\providecommand \translation [1]{[#1]}%
\providecommand \BibitemOpen [0]{}%
\providecommand \bibitemStop [0]{}%
\providecommand \bibitemNoStop [0]{.\EOS\space}%
\providecommand \EOS [0]{\spacefactor3000\relax}%
\providecommand \BibitemShut  [1]{\csname bibitem#1\endcsname}%
\let\auto@bib@innerbib\@empty
\bibitem [{\citenamefont {Walls}(1983)}]{walls1983squeezed}%
  \BibitemOpen
  \bibfield  {author} {\bibinfo {author} {\bibfnamefont {D.~F.}\ \bibnamefont
  {Walls}},\ }\href@noop {} {\bibfield  {journal} {\bibinfo  {journal}
  {Nature}\ }\textbf {\bibinfo {volume} {306}},\ \bibinfo {pages} {141}
  (\bibinfo {year} {1983})}\BibitemShut {NoStop}%
\bibitem [{\citenamefont {Wu}\ \emph {et~al.}(1986)\citenamefont {Wu},
  \citenamefont {Kimble}, \citenamefont {Hall},\ and\ \citenamefont
  {Wu}}]{PhysRevLett.57.2520}%
  \BibitemOpen
  \bibfield  {author} {\bibinfo {author} {\bibfnamefont {L.~A.}\ \bibnamefont
  {Wu}}, \bibinfo {author} {\bibfnamefont {H.~J.}\ \bibnamefont {Kimble}},
  \bibinfo {author} {\bibfnamefont {J.~L.}\ \bibnamefont {Hall}}, \ and\
  \bibinfo {author} {\bibfnamefont {H.}~\bibnamefont {Wu}},\ }\href@noop {}
  {\bibfield  {journal} {\bibinfo  {journal} {Phys. Rev. Lett.}\ }\textbf
  {\bibinfo {volume} {57}},\ \bibinfo {pages} {2520} (\bibinfo {year}
  {1986})}\BibitemShut {NoStop}%
\bibitem [{\citenamefont {Reid}\ and\ \citenamefont
  {Drummond}(1988)}]{PhysRevLett.60.2731}%
  \BibitemOpen
  \bibfield  {author} {\bibinfo {author} {\bibfnamefont {M.~D.}\ \bibnamefont
  {Reid}}\ and\ \bibinfo {author} {\bibfnamefont {P.~D.}\ \bibnamefont
  {Drummond}},\ }\href@noop {} {\bibfield  {journal} {\bibinfo  {journal}
  {Phys. Rev. Lett.}\ }\textbf {\bibinfo {volume} {60}},\ \bibinfo {pages}
  {2731} (\bibinfo {year} {1988})}\BibitemShut {NoStop}%
\bibitem [{\citenamefont {Janousek}\ \emph {et~al.}(2009)\citenamefont
  {Janousek}, \citenamefont {Wagner}, \citenamefont {Morizur}, \citenamefont
  {Treps}, \citenamefont {Lam}, \citenamefont {Harb},\ and\ \citenamefont
  {Bachor}}]{janousek2009optical}%
  \BibitemOpen
  \bibfield  {author} {\bibinfo {author} {\bibfnamefont {J.}~\bibnamefont
  {Janousek}}, \bibinfo {author} {\bibfnamefont {K.}~\bibnamefont {Wagner}},
  \bibinfo {author} {\bibfnamefont {J.}~\bibnamefont {Morizur}}, \bibinfo
  {author} {\bibfnamefont {N.}~\bibnamefont {Treps}}, \bibinfo {author}
  {\bibfnamefont {P.}~\bibnamefont {Lam}}, \bibinfo {author} {\bibfnamefont
  {C.}~\bibnamefont {Harb}}, \ and\ \bibinfo {author} {\bibfnamefont
  {H.}~\bibnamefont {Bachor}},\ }\href@noop {} {\bibfield  {journal} {\bibinfo
  {journal} {Nature Photonics}\ }\textbf {\bibinfo {volume} {3}},\ \bibinfo
  {pages} {399} (\bibinfo {year} {2009})}\BibitemShut {NoStop}%
\bibitem [{\citenamefont {Bouwmeester}\ \emph {et~al.}(1997)\citenamefont
  {Bouwmeester}, \citenamefont {Pan}, \citenamefont {Mattle}, \citenamefont
  {Eibl}, \citenamefont {Weinfurter},\ and\ \citenamefont
  {Zeilinger}}]{bouwmeester1997experimental}%
  \BibitemOpen
  \bibfield  {author} {\bibinfo {author} {\bibfnamefont {D.}~\bibnamefont
  {Bouwmeester}}, \bibinfo {author} {\bibfnamefont {J.-W.}\ \bibnamefont
  {Pan}}, \bibinfo {author} {\bibfnamefont {K.}~\bibnamefont {Mattle}},
  \bibinfo {author} {\bibfnamefont {M.}~\bibnamefont {Eibl}}, \bibinfo {author}
  {\bibfnamefont {H.}~\bibnamefont {Weinfurter}}, \ and\ \bibinfo {author}
  {\bibfnamefont {A.}~\bibnamefont {Zeilinger}},\ }\href@noop {} {\bibfield
  {journal} {\bibinfo  {journal} {Nature}\ }\textbf {\bibinfo {volume} {390}},\
  \bibinfo {pages} {575} (\bibinfo {year} {1997})}\BibitemShut {NoStop}%
\bibitem [{\citenamefont {Furusawa}\ \emph {et~al.}(1998)\citenamefont
  {Furusawa}, \citenamefont {S{\o}rensen}, \citenamefont {Braunstein},
  \citenamefont {Fuchs}, \citenamefont {Kimble},\ and\ \citenamefont
  {Polzik}}]{furusawa1998unconditional}%
  \BibitemOpen
  \bibfield  {author} {\bibinfo {author} {\bibfnamefont {A.}~\bibnamefont
  {Furusawa}}, \bibinfo {author} {\bibfnamefont {J.~L.}\ \bibnamefont
  {S{\o}rensen}}, \bibinfo {author} {\bibfnamefont {S.~L.}\ \bibnamefont
  {Braunstein}}, \bibinfo {author} {\bibfnamefont {C.~A.}\ \bibnamefont
  {Fuchs}}, \bibinfo {author} {\bibfnamefont {H.~J.}\ \bibnamefont {Kimble}}, \
  and\ \bibinfo {author} {\bibfnamefont {E.~S.}\ \bibnamefont {Polzik}},\
  }\href@noop {} {\bibfield  {journal} {\bibinfo  {journal} {Science}\ }\textbf
  {\bibinfo {volume} {282}},\ \bibinfo {pages} {706} (\bibinfo {year}
  {1998})}\BibitemShut {NoStop}%
\bibitem [{\citenamefont {Wong}\ \emph {et~al.}(2010)\citenamefont {Wong},
  \citenamefont {Vodopyanov},\ and\ \citenamefont {Byer}}]{Wong:10}%
  \BibitemOpen
  \bibfield  {author} {\bibinfo {author} {\bibfnamefont {S.~T.}\ \bibnamefont
  {Wong}}, \bibinfo {author} {\bibfnamefont {K.~L.}\ \bibnamefont
  {Vodopyanov}}, \ and\ \bibinfo {author} {\bibfnamefont {R.~L.}\ \bibnamefont
  {Byer}},\ }\href@noop {} {\bibfield  {journal} {\bibinfo  {journal} {J. Opt.
  Soc. Am. B}\ }\textbf {\bibinfo {volume} {27}},\ \bibinfo {pages} {876}
  (\bibinfo {year} {2010})}\BibitemShut {NoStop}%
\bibitem [{\citenamefont {Crisafulli}\ \emph {et~al.}(2013)\citenamefont
  {Crisafulli}, \citenamefont {Tezak}, \citenamefont {Soh}, \citenamefont
  {Armen},\ and\ \citenamefont {Mabuchi}}]{Crisafulli:13}%
  \BibitemOpen
  \bibfield  {author} {\bibinfo {author} {\bibfnamefont {O.}~\bibnamefont
  {Crisafulli}}, \bibinfo {author} {\bibfnamefont {N.}~\bibnamefont {Tezak}},
  \bibinfo {author} {\bibfnamefont {D.~B.~S.}\ \bibnamefont {Soh}}, \bibinfo
  {author} {\bibfnamefont {M.~A.}\ \bibnamefont {Armen}}, \ and\ \bibinfo
  {author} {\bibfnamefont {H.}~\bibnamefont {Mabuchi}},\ }\href@noop {}
  {\bibfield  {journal} {\bibinfo  {journal} {Opt. Express}\ }\textbf {\bibinfo
  {volume} {21}},\ \bibinfo {pages} {18371} (\bibinfo {year}
  {2013})}\BibitemShut {NoStop}%
\bibitem [{\citenamefont {Yokoyama}\ \emph {et~al.}(2013)\citenamefont
  {Yokoyama}, \citenamefont {Ukai}, \citenamefont {Armstrong}, \citenamefont
  {Sornphiphatphong}, \citenamefont {Kaji}, \citenamefont {Suzuki},
  \citenamefont {Yoshikawa}, \citenamefont {Yonezawa}, \citenamefont
  {Menicucci},\ and\ \citenamefont {Furusawa}}]{yokoyama2013ultra}%
  \BibitemOpen
  \bibfield  {author} {\bibinfo {author} {\bibfnamefont {S.}~\bibnamefont
  {Yokoyama}}, \bibinfo {author} {\bibfnamefont {R.}~\bibnamefont {Ukai}},
  \bibinfo {author} {\bibfnamefont {S.~C.}\ \bibnamefont {Armstrong}}, \bibinfo
  {author} {\bibfnamefont {C.}~\bibnamefont {Sornphiphatphong}}, \bibinfo
  {author} {\bibfnamefont {T.}~\bibnamefont {Kaji}}, \bibinfo {author}
  {\bibfnamefont {S.}~\bibnamefont {Suzuki}}, \bibinfo {author} {\bibfnamefont
  {J.-i.}\ \bibnamefont {Yoshikawa}}, \bibinfo {author} {\bibfnamefont
  {H.}~\bibnamefont {Yonezawa}}, \bibinfo {author} {\bibfnamefont {N.~C.}\
  \bibnamefont {Menicucci}}, \ and\ \bibinfo {author} {\bibfnamefont
  {A.}~\bibnamefont {Furusawa}},\ }\href@noop {} {\bibfield  {journal}
  {\bibinfo  {journal} {Nature Photonics}\ }\textbf {\bibinfo {volume} {7}},\
  \bibinfo {pages} {982} (\bibinfo {year} {2013})}\BibitemShut {NoStop}%
\bibitem [{\citenamefont {Wang}\ \emph {et~al.}(2013)\citenamefont {Wang},
  \citenamefont {Marandi}, \citenamefont {Wen}, \citenamefont {Byer},\ and\
  \citenamefont {Yamamoto}}]{PhysRevA.88.063853}%
  \BibitemOpen
  \bibfield  {author} {\bibinfo {author} {\bibfnamefont {Z.}~\bibnamefont
  {Wang}}, \bibinfo {author} {\bibfnamefont {A.}~\bibnamefont {Marandi}},
  \bibinfo {author} {\bibfnamefont {K.}~\bibnamefont {Wen}}, \bibinfo {author}
  {\bibfnamefont {R.~L.}\ \bibnamefont {Byer}}, \ and\ \bibinfo {author}
  {\bibfnamefont {Y.}~\bibnamefont {Yamamoto}},\ }\href@noop {} {\bibfield
  {journal} {\bibinfo  {journal} {Phys. Rev. A}\ }\textbf {\bibinfo {volume}
  {88}},\ \bibinfo {pages} {063853} (\bibinfo {year} {2013})}\BibitemShut
  {NoStop}%
\bibitem [{\citenamefont {Marandi}\ \emph {et~al.}(2014)\citenamefont
  {Marandi}, \citenamefont {Wang}, \citenamefont {Takata}, \citenamefont
  {Byer},\ and\ \citenamefont {Yamamoto}}]{marandi2014network}%
  \BibitemOpen
  \bibfield  {author} {\bibinfo {author} {\bibfnamefont {A.}~\bibnamefont
  {Marandi}}, \bibinfo {author} {\bibfnamefont {Z.}~\bibnamefont {Wang}},
  \bibinfo {author} {\bibfnamefont {K.}~\bibnamefont {Takata}}, \bibinfo
  {author} {\bibfnamefont {R.~L.}\ \bibnamefont {Byer}}, \ and\ \bibinfo
  {author} {\bibfnamefont {Y.}~\bibnamefont {Yamamoto}},\ }\href@noop {}
  {\bibfield  {journal} {\bibinfo  {journal} {Nature Photonics}\ }\textbf
  {\bibinfo {volume} {8}},\ \bibinfo {pages} {931} (\bibinfo {year}
  {2014})}\BibitemShut {NoStop}%
\bibitem [{\citenamefont {Wolinsky}\ and\ \citenamefont
  {Carmichael}(1988)}]{wolinsky1988quantum}%
  \BibitemOpen
  \bibfield  {author} {\bibinfo {author} {\bibfnamefont {M.}~\bibnamefont
  {Wolinsky}}\ and\ \bibinfo {author} {\bibfnamefont {H.}~\bibnamefont
  {Carmichael}},\ }\href@noop {} {\bibfield  {journal} {\bibinfo  {journal}
  {Physical Review Letters}\ }\textbf {\bibinfo {volume} {60}},\ \bibinfo
  {pages} {1836} (\bibinfo {year} {1988})}\BibitemShut {NoStop}%
\bibitem [{\citenamefont {Krippner}\ \emph {et~al.}(1994)\citenamefont
  {Krippner}, \citenamefont {Munro},\ and\ \citenamefont
  {Reid}}]{krippner1994transient}%
  \BibitemOpen
  \bibfield  {author} {\bibinfo {author} {\bibfnamefont {L.}~\bibnamefont
  {Krippner}}, \bibinfo {author} {\bibfnamefont {W.}~\bibnamefont {Munro}}, \
  and\ \bibinfo {author} {\bibfnamefont {M.}~\bibnamefont {Reid}},\ }\href@noop
  {} {\bibfield  {journal} {\bibinfo  {journal} {Physical Review A}\ }\textbf
  {\bibinfo {volume} {50}},\ \bibinfo {pages} {4330} (\bibinfo {year}
  {1994})}\BibitemShut {NoStop}%
\bibitem [{\citenamefont {Takata}\ and\ \citenamefont
  {Yamamoto}(2015)}]{Kenta}%
  \BibitemOpen
  \bibfield  {author} {\bibinfo {author} {\bibfnamefont {K.}~\bibnamefont
  {Takata}}\ and\ \bibinfo {author} {\bibfnamefont {Y.}~\bibnamefont
  {Yamamoto}},\ }\href@noop {} {\bibfield  {journal} {\bibinfo  {journal}
  {Phys. Rev. A}\ }\textbf {\bibinfo {volume} {92}},\ \bibinfo {pages} {043821}
  (\bibinfo {year} {2015})}\BibitemShut {NoStop}%
\bibitem [{\citenamefont {Haribara}\ \emph {et~al.}(2016)\citenamefont
  {Haribara}, \citenamefont {Utsunomiya},\ and\ \citenamefont
  {Yamamoto}}]{haribara2016computational}%
  \BibitemOpen
  \bibfield  {author} {\bibinfo {author} {\bibfnamefont {Y.}~\bibnamefont
  {Haribara}}, \bibinfo {author} {\bibfnamefont {S.}~\bibnamefont
  {Utsunomiya}}, \ and\ \bibinfo {author} {\bibfnamefont {Y.}~\bibnamefont
  {Yamamoto}},\ }\href@noop {} {\bibfield  {journal} {\bibinfo  {journal}
  {Entropy}\ }\textbf {\bibinfo {volume} {18}},\ \bibinfo {pages} {151}
  (\bibinfo {year} {2016})}\BibitemShut {NoStop}%
\bibitem [{\citenamefont {Munro}\ and\ \citenamefont
  {Reid}(1995)}]{munro1995transient}%
  \BibitemOpen
  \bibfield  {author} {\bibinfo {author} {\bibfnamefont {W.}~\bibnamefont
  {Munro}}\ and\ \bibinfo {author} {\bibfnamefont {M.}~\bibnamefont {Reid}},\
  }\href@noop {} {\bibfield  {journal} {\bibinfo  {journal} {Physical Review
  A}\ }\textbf {\bibinfo {volume} {52}},\ \bibinfo {pages} {2388} (\bibinfo
  {year} {1995})}\BibitemShut {NoStop}%
\bibitem [{\citenamefont {Carmichael}(1999)}]{carmichael1999statistical}%
  \BibitemOpen
  \bibfield  {author} {\bibinfo {author} {\bibfnamefont {H.}~\bibnamefont
  {Carmichael}},\ }\href@noop {} {\emph {\bibinfo {title} {Statistical Methods
  in Quantum Optics 1: Master Equations and Fokker-Planck Equations}}}\
  (\bibinfo  {publisher} {Springer},\ \bibinfo {year} {1999})\BibitemShut
  {NoStop}%
\bibitem [{\citenamefont {Drummond}\ and\ \citenamefont
  {Gardiner}(1980)}]{PPrep}%
  \BibitemOpen
  \bibfield  {author} {\bibinfo {author} {\bibfnamefont {P.~D.}\ \bibnamefont
  {Drummond}}\ and\ \bibinfo {author} {\bibfnamefont {C.~W.}\ \bibnamefont
  {Gardiner}},\ }\href@noop {} {\bibfield  {journal} {\bibinfo  {journal}
  {Journal of Physics A: Mathematical and General}\ }\textbf {\bibinfo {volume}
  {13}},\ \bibinfo {pages} {2353} (\bibinfo {year} {1980})}\BibitemShut
  {NoStop}%
\bibitem [{\citenamefont {Glauber}(1963)}]{PhysRev.131.2766}%
  \BibitemOpen
  \bibfield  {author} {\bibinfo {author} {\bibfnamefont {R.~J.}\ \bibnamefont
  {Glauber}},\ }\href@noop {} {\bibfield  {journal} {\bibinfo  {journal} {Phys.
  Rev.}\ }\textbf {\bibinfo {volume} {131}},\ \bibinfo {pages} {2766} (\bibinfo
  {year} {1963})}\BibitemShut {NoStop}%
\bibitem [{\citenamefont {Walls}\ and\ \citenamefont
  {Milburn}(2007)}]{walls2007quantum}%
  \BibitemOpen
  \bibfield  {author} {\bibinfo {author} {\bibfnamefont {D.~F.}\ \bibnamefont
  {Walls}}\ and\ \bibinfo {author} {\bibfnamefont {G.~J.}\ \bibnamefont
  {Milburn}},\ }\href@noop {} {\emph {\bibinfo {title} {Quantum Optics}}}\
  (\bibinfo  {publisher} {Springer Science \& Business Media},\ \bibinfo {year}
  {2007})\BibitemShut {NoStop}%
\bibitem [{\citenamefont {Duan}\ \emph {et~al.}(2000)\citenamefont {Duan},
  \citenamefont {Giedke}, \citenamefont {Cirac},\ and\ \citenamefont
  {Zoller}}]{PhysRevLett.84.2722}%
  \BibitemOpen
  \bibfield  {author} {\bibinfo {author} {\bibfnamefont {L.~M.}\ \bibnamefont
  {Duan}}, \bibinfo {author} {\bibfnamefont {G.}~\bibnamefont {Giedke}},
  \bibinfo {author} {\bibfnamefont {J.~I.}\ \bibnamefont {Cirac}}, \ and\
  \bibinfo {author} {\bibfnamefont {P.}~\bibnamefont {Zoller}},\ }\href@noop {}
  {\bibfield  {journal} {\bibinfo  {journal} {Phys. Rev. Lett.}\ }\textbf
  {\bibinfo {volume} {84}},\ \bibinfo {pages} {2722} (\bibinfo {year}
  {2000})}\BibitemShut {NoStop}%
\bibitem [{\citenamefont {Drummond}\ \emph {et~al.}(2002)\citenamefont
  {Drummond}, \citenamefont {Dechoum},\ and\ \citenamefont
  {Chaturvedi}}]{drummond2002critical}%
  \BibitemOpen
  \bibfield  {author} {\bibinfo {author} {\bibfnamefont {P.}~\bibnamefont
  {Drummond}}, \bibinfo {author} {\bibfnamefont {K.}~\bibnamefont {Dechoum}}, \
  and\ \bibinfo {author} {\bibfnamefont {S.}~\bibnamefont {Chaturvedi}},\
  }\href@noop {} {\bibfield  {journal} {\bibinfo  {journal} {Physical Review
  A}\ }\textbf {\bibinfo {volume} {65}},\ \bibinfo {pages} {033806} (\bibinfo
  {year} {2002})}\BibitemShut {NoStop}%
\bibitem [{\citenamefont {Chaturvedi}\ \emph {et~al.}(2002)\citenamefont
  {Chaturvedi}, \citenamefont {Dechoum},\ and\ \citenamefont
  {Drummond}}]{chaturvedi2002limits}%
  \BibitemOpen
  \bibfield  {author} {\bibinfo {author} {\bibfnamefont {S.}~\bibnamefont
  {Chaturvedi}}, \bibinfo {author} {\bibfnamefont {K.}~\bibnamefont {Dechoum}},
  \ and\ \bibinfo {author} {\bibfnamefont {P.}~\bibnamefont {Drummond}},\
  }\href@noop {} {\bibfield  {journal} {\bibinfo  {journal} {Physical Review
  A}\ }\textbf {\bibinfo {volume} {65}},\ \bibinfo {pages} {033805} (\bibinfo
  {year} {2002})}\BibitemShut {NoStop}%
\bibitem [{\citenamefont {Dechoum}\ \emph {et~al.}(2004)\citenamefont
  {Dechoum}, \citenamefont {Drummond}, \citenamefont {Chaturvedi},\ and\
  \citenamefont {Reid}}]{dechoum2004critical}%
  \BibitemOpen
  \bibfield  {author} {\bibinfo {author} {\bibfnamefont {K.}~\bibnamefont
  {Dechoum}}, \bibinfo {author} {\bibfnamefont {P.}~\bibnamefont {Drummond}},
  \bibinfo {author} {\bibfnamefont {S.}~\bibnamefont {Chaturvedi}}, \ and\
  \bibinfo {author} {\bibfnamefont {M.}~\bibnamefont {Reid}},\ }\href@noop {}
  {\bibfield  {journal} {\bibinfo  {journal} {Physical Review A}\ }\textbf
  {\bibinfo {volume} {70}},\ \bibinfo {pages} {053807} (\bibinfo {year}
  {2004})}\BibitemShut {NoStop}%
\bibitem [{\citenamefont {Yuen}\ and\ \citenamefont {Chan}(1983)}]{Yuen:83}%
  \BibitemOpen
  \bibfield  {author} {\bibinfo {author} {\bibfnamefont {H.~P.}\ \bibnamefont
  {Yuen}}\ and\ \bibinfo {author} {\bibfnamefont {V.~W.~S.}\ \bibnamefont
  {Chan}},\ }\href@noop {} {\bibfield  {journal} {\bibinfo  {journal} {Opt.
  Lett.}\ }\textbf {\bibinfo {volume} {8}},\ \bibinfo {pages} {177} (\bibinfo
  {year} {1983})}\BibitemShut {NoStop}%
\bibitem [{\citenamefont {La~Porta}\ \emph {et~al.}(1989)\citenamefont
  {La~Porta}, \citenamefont {Slusher},\ and\ \citenamefont
  {Yurke}}]{PhysRevLett.62.28}%
  \BibitemOpen
  \bibfield  {author} {\bibinfo {author} {\bibfnamefont {A.}~\bibnamefont
  {La~Porta}}, \bibinfo {author} {\bibfnamefont {R.~E.}\ \bibnamefont
  {Slusher}}, \ and\ \bibinfo {author} {\bibfnamefont {B.}~\bibnamefont
  {Yurke}},\ }\href@noop {} {\bibfield  {journal} {\bibinfo  {journal} {Phys.
  Rev. Lett.}\ }\textbf {\bibinfo {volume} {62}},\ \bibinfo {pages} {28}
  (\bibinfo {year} {1989})}\BibitemShut {NoStop}%
\end{thebibliography}%

\end{document}